\documentclass[journal]{IEEEtran}

\ifCLASSINFOpdf
\else
\fi

\usepackage{subfig}
\usepackage{graphicx}
\usepackage[english]{babel}
\usepackage{amsthm}
\usepackage{amsmath}
\usepackage{amsfonts}
\usepackage{amssymb}
\usepackage{framed}
\usepackage{bm}
\usepackage{tikz}
\usetikzlibrary{arrows}
\usepackage {tikz}
\usetikzlibrary {positioning}
\usepackage{changebar}
\definecolor {processblue}{cmyk}{0.96,0,0,0}

\newtheorem{prop}{Proposition}
\newtheorem{thm}{Theorem}
\newtheorem{rmk}{Remark}
\newtheorem{lemma}{Lemma}
\newtheorem{definition}{Definition}
\newtheorem{exe}{Example}
\newtheorem{cor}{Corollary}
\newcommand{\vol}{\mbox{\upshape{vol} }}
\newcommand{\rep}{\text{\upshape{rep} }}
\newcommand{\trep}{t_{\text{\upshape{rep} }}}

\newif\ifmnote
\mnotefalse

\newif\iftodolist
\todolistfalse

\begin{document}

\title{\sc Reliability of Erasure Coded Storage Systems: A Combinatorial-Geometric Approach\thanks{This work was partially presented at the IEEE International Conference on BigData 2013, and  at the IEEE Information Theory Workshop 2014, Tasmania, Australia.}}

\author{Vinay~A.~Vaishampayan, \textit{Fellow, IEEE}\thanks{Vinay A. Vaishampayan, formerly with AT\&T's Shannon Laboratory, is now with the City University of New York, College of Staten Island, Department of Engineering Science and Physics. His work was supported in part by CNPq Grant 400441/2014-4 and PSC-CUNY Award \# 68631-00 46} and Antonio~Campello, \textit{Member, IEEE}\thanks{Antonio Campello is currently with T\'el\'ecom-ParisTech, France and University of Campinas, Brazil. His work was supported by S\~ao Paulo Research Foundation (FAPESP) grants 2013/25219-5, 2014/20602-8 and was initiated during a short-term visit to AT\&T Shannon Laboratory, NJ, USA.} \thanks{Copyright (c) 2014 IEEE. Personal use of this material is permitted.  However, permission to use this material for any other purposes must be obtained from the IEEE by sending a request to pubs-permissions@ieee.org.}}

\maketitle

\begin{abstract}
We consider the probability of data loss, or equivalently, the reliability function for an erasure coded  distributed data storage system under worst case conditions. Data loss in an erasure coded system depends on probability distributions for the disk repair duration and the disk failure duration.  In previous works, the data loss probability of such systems has  been studied under the assumption of exponentially distributed disk failure and disk repair durations, using well-known analytic methods from the theory of Markov processes. These methods lead to an estimate of the integral of the reliability function.

Here, we  address the problem of directly calculating the data loss probability for general repair and failure duration distributions. A  closed limiting form is developed for the probability of data loss and it is shown that the probability of the event that a repair duration exceeds a failure duration is sufficient for characterizing the data loss probability. 

For the case of constant repair duration, we develop an expression for the conditional data loss probability given the number of failures experienced by a each node in a given time window. We do so by developing a geometric approach that relies on the computation of volumes of a family of polytopes that are related to the code. An exact calculation is provided and an upper bound on the data loss probability is obtained by posing the problem as a set avoidance problem.  Theoretical calculations are compared to simulation results. 
\end{abstract}


\section{Introduction}
\label{sec:Introduction}
Distributed data storage systems are growing in popularity, driven by demand and enabled by the availability of broadband networks, and declining costs of storage devices. Erasure coding represents a practical method for building highly reliable storage systems using low cost, less reliable storage drives. In an erasure coded storage system, a block of $k$ information symbols from some finite set
is encoded into a block of $n$ coded symbols by an $(n,k)$ erasure code and the $n$ code symbols are placed on separate disks.
When a disk fails, it is {\it repaired}, i.e.  redundant information in the code is used to recompute the erased symbol which is then 
placed on a replacement disk. Repair is essential for the reliability of the overall
system. Data loss occurs or the system fails when the total number of failed disks at any time exceeds the erasure correcting capability of the code.
If disks are repaired swiftly, the number of failed disks can be kept small on average, reducing the probability of data loss. 

 An important metric is the reliability function $R(t)$, defined to be the probability that data is not lost in the time window $[0,t]$. In previous works~\cite{Angus:1988}, \cite{Chen:1994}, it is assumed that the repair and failure durations are exponentially distributed random variables and the mean time to data loss (MTTDL) is  determined by analyzing a state transition diagram, where the system state at a given time is defined as the number of working disks at that time, see e.g.~\cite{resch2011development}, \cite{Angus:1988}, \cite{Chen:1994}. The reliability function $R(t)$ is then estimated by  the formula $\exp (-t/\mbox{MTTDL})$.  Exponentially distributed and independent  failure durations are critical for this analysis to proceed. 
Several disk failure and disk repair modeling and measurement studies have been reported in the literature, e.g. ~\cite{schroeder2007disk},~\cite{pinheiro2007failure}, \cite{elerath2007enhanced}, \cite{xin2003reliability}. It is concluded that real world storage devices do not exhibit exponentially distributed lifetimes and that the Weibull distribution with appropriately chosen parameters is a  more appropriate model for failure and repair durations. Simulation is a valuable tool for evaluating reliability of disk storage systems, see e.g. \cite{greenan2009reliability} which also includes a comprehensive review of previous modeling studies. 
In a recent contribution~\cite{VenkThesis2012}, it is shown that the reliability analysis based on the above exponential model is robust to changes in the disk failure time distribution. It is worth noting that~\cite{VenkThesis2012}  also points out that the analysis of MTTDL is \emph{not} robust to changes in the repair duration distribution.

The main contributions of this work  are summarized below.
\begin{enumerate}
\item We derive a formula for the data loss probability $P({\mathcal D}_t)$ for small $G$ and large $t$, for general independent and identically distributed (iid) failure and repair distributions, (\ref{eq:LimitingForm}), restated here for convenience
\begin{equation}
\frac{  P({\mathcal D}_t) }{(G/n)^{(n-k)} } \approx  \frac{(n-1)!}{(k-1)!} \frac{t}{E(Y)},
\end{equation}
where random variable $Y$ represents a failure duration, $G$ is the probability that $Y<Z$, where $Z$ is a random repair duration and an $(n,k)$ MDS erasure code is used.
This is obtained by conditioning on a specific sequence of binary events, to be described later. Our derivation shows that
for general distributions, the data loss probability is characterized in terms of the probability that the failure duration is smaller than the repair duration, and supports the fact (already known in the literature) that failure and repair rates are insufficient characterizations for determining the data loss probability. Our contribution is to show that the above mentioned probability is {\it sufficient} for characterizing the limiting data loss probability.

\item For constant repair duration, by conditioning on the number of failure events on each node, we arrive at another expression for the data loss probability, as well as a lower bound i.e.  we derive an expression for $P_{\bm{m}}({\mathcal D}_t)$, the data loss probability conditioned on $\bm{m}=(m_{1},m_{2},\ldots,m_{n})$, where $m_{i}$ is the number of failures for disk $i$ in time window $[0,t]$. Specifically, we prove that $P_{\bm{m}}(\mathcal{D}_t)$ has the following asymptotic behavior as  $\tau = t_\rep/t \to 0$: 
\begin{eqnarray}
\lefteqn{\lim_{\tau\to 0} \frac{P_{\bm{m}}(D_t)}{\tau^{n-k}} } \nonumber \\
& = (n-k+1)! \displaystyle \sum_{(i_1,\ldots,i_{n-k+1}) \above 0pt \mbox{{\tiny distinct}}}{m_{i_1} m_{i_2}\ldots m_{i_{n-k+1}}}. \nonumber \\
\end{eqnarray}
This analysis  holds for exponential failure distributions and provides a finer analysis of the system, not addressed by previous Markov chain approaches. It also has some implications to non-homogeneous Poisson processes.

\item By viewing the data loss probability calculation for constant repair duration as a problem of set avoidance by the Cartesian product of random sets, we derive an \emph{upper} bound on the data loss probability,
\begin{equation} 
\begin{split} P_{\bm{m}}(\mathcal{D}_t) \leq 1 -\left(1-\frac{\mbox{\upshape{vol }}\mathcal{R}}{t^n}\right)^{m_1 \ldots m_n},
\end{split}
\end{equation}
where $\mathcal{R} \subset [0,t]^n$ is a suitably defined error region associated with the code. The simplicity of the upper bound and the fact that its asymptotic behavior is comparable to the closed forms in some regimes makes it useful in practice. Methods for sharpening this bound remain as an open question.
\item  We explore the connection between the erasure code and a family of polytopes that determine the error region. This connection is, in our opinion, interesting in its own right, even though it comes from an error probability calculation. Our contribution here is to develop a systematic  approach for calculating the volume of a set of ordered points with constrained differences between successive elements. This method underlies the calculations  for constant repair duration in this paper.
\end{enumerate}

The paper is organized as follows. Sec.~\ref{sec:Statement} contains a problem statement and states the  assumptions that underlie our analysis. The data loss probability for general distributions is derived in Sec.~\ref{sec-anal-gen}.
For constant repair durations, we explore the combinatorial and geometric aspects  of the problem of evaluating the data loss probability in Sec.~\ref{sec:constant}.  Sec.~\ref{sec:avoidance} presents a method for upper bounding the data loss probability for constant repair duration by viewing the problem as a set avoidance problem. 
Volume calculations that underlie both the direct calculation as well as the set avoidance upper bound are presented in Sec.~\ref{sec:ErrorRegions}. Numerical and simulation results that explore some of the implications of the theory developed are presented in Sec.~\ref{sec:Numerical}. The paper is summarized and suggestions for future research are presented in Sec.~\ref{sec:Conclusions}. Some mathematical details and a proofs are contained in the appendix.

\section{Assumptions, Problem Statement and an Example}
\label{sec:Statement}
Code symbols from an MDS $(n,k)$ erasure code are written to $n$ disks\footnote{To be precise, in the modern terminology it is said that the information is stored in a \textit{node}. Throughout the paper we use the looser term disk instead, in analogy to classical storage systems.}. We assume  that disk failures occur independently and that the disk failure process is modeled by an independent increment process with known probability distribution.  


When a disk fails, data is downloaded from other disks and used to repair the lost symbols on the failed disk. We refer to these disks as {\it helper} disks, and to the set of helper disks as the {\it helper set}. The probability distribution of the repair duration, $Z$, is known, and repair durations are assumed to be independent and identically distributed. Since the codes are MDS, we consider that data is available as long as at least $k$ disks are working (alternatively, if there was no instant of time at which less than $k$ disks were working).
Thus a data loss event occurs in the interval $[0,t)$ if the number of failed disks exceeds $(n-k)$ the erasure correcting capability
of the code.

Characterization of a data loss event is subtle and depends on the system architecture, as well as on characteristics of the erasure code.
An example is shown in  Fig.~\ref{fig:ErrorEventParallel} for constant repair duration $\trep$. Disk 1 has failed and the helper set consists of disks 3 and 4. However, prior to disk 1 being restored, disk 2 fails. With a traditional MDS code, replacement symbols for disk 2 would be computed and the repair of disk 2 would begin without
interrupting the repair of disk 1. On the other hand, in systems that perform functional repair~\cite{dimakis2010network}, it is possible that the symbols for disk 1 would need to be
recomputed as well, which implies that the repair process for disk 1 would need to be restarted. As a consequence, this sequence of failures and repairs results in a data loss event. 

\begin{framed}{In our analysis we consider a disk to be repaired if and only if that disk repairs successfully, or a subsequently failed disk repairs successfully before the total number of failed disks exceeds the erasure correcting capability of the code. }
\end{framed}

\begin{figure}[htbp] 
   \centering
   \includegraphics[width=3in]{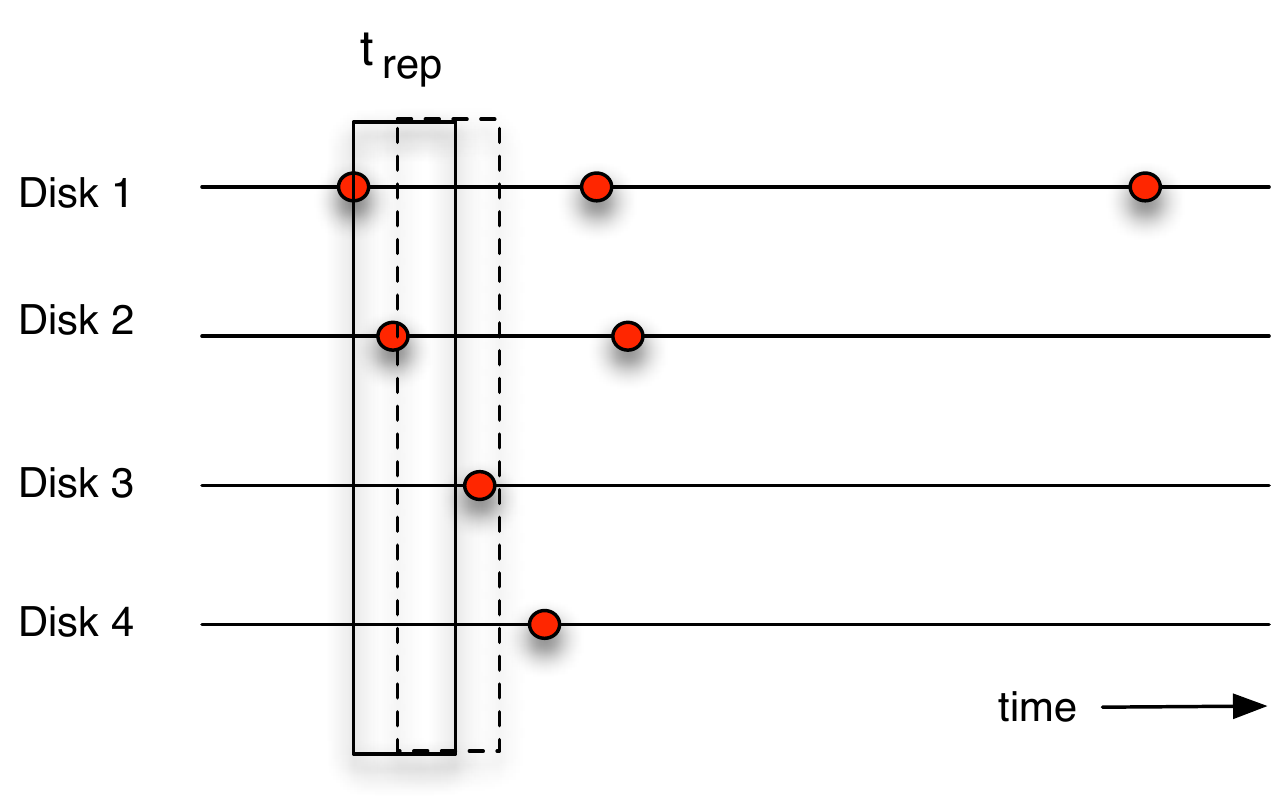} 
   \caption{Sequence of disk failures (shaded dots) that causes a data loss for a $(4,2)$ coded system and helper set of size $2$.  Here the helper set for disk $1$  is $\{3,4\}$. When disk 2 fails, even though the helper set remains unchanged, the symbols for disk 1 must be recomputed (for functional repair).  Since  disk 3 fails prior to the repair being completed, a data loss event has occurred.}
   \label{fig:ErrorEventParallel}
\end{figure}

\section{Analysis for General Failure and Repair Time Distributions}
\label{sec-anal-gen}
\begin{figure}[h] 
   \centering
   \includegraphics[width=3in]{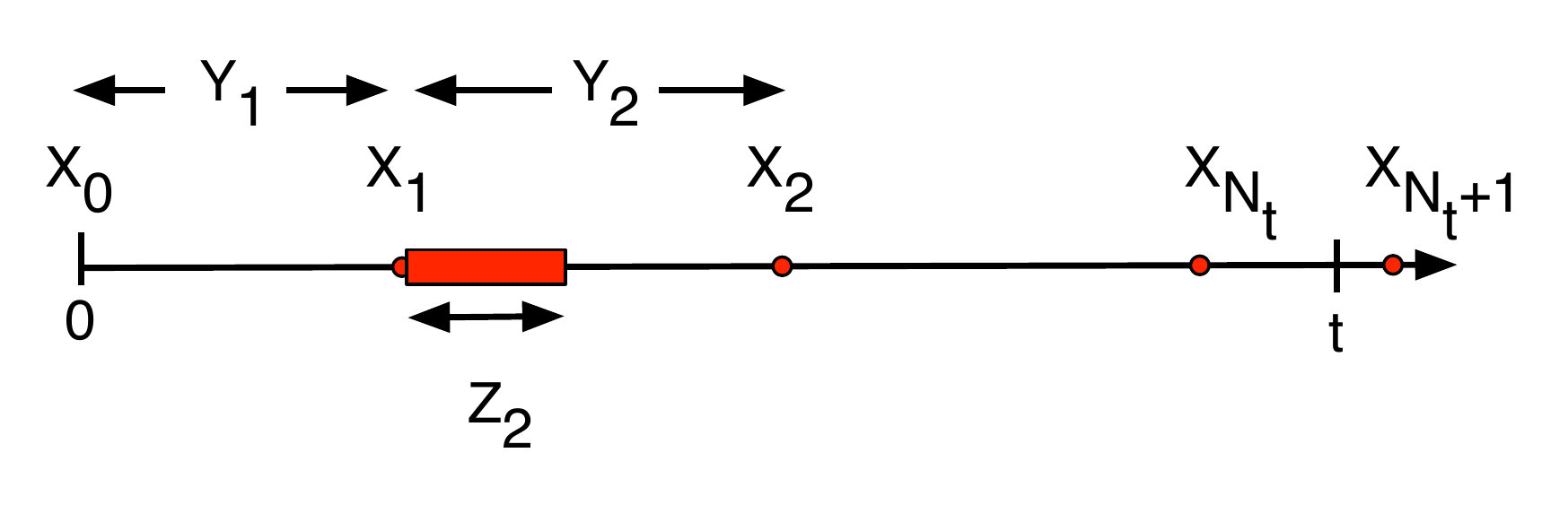} 
   \caption{Illustration of setup for failure times and repair durations.}
   \label{fig:anal-gen}
\end{figure}
We make the following assumption about our failure process. The $i$th inter-failure duration (hereafter referred to as the failure duration) for the \emph{system} is denoted $Y_i$.  The process $\{Y_i,~i=1,2,\ldots,\}$ is an i.i.d. process with known probability density function (pdf) $f_Y(\cdot)$, cumulative distributive function (CDF) $F_Y(\cdot)$, where $f_Y(y)=0$ for $y < 0$. Let $X_0=0$ and $X_i=X_{i-1}+Y_{i}$, $i=1,2,\ldots$. Here $X_i,~i=1,2,\ldots,$ is the instant at which the $i$th disk failure in the system occurs.  Note that $X_0$ is not regarded as a failure instant. Let the random variable $N_t$ count the number of failures in time interval $[0,t)$, i.e. $X_{N_t} <t$ and $X_{N_t+1}\geq t$. For an $(n,k)$ MDS code an error cannot occur in $[0,t)$ if $N_t < n-k+1$.

Our failure process associates with each $X_i$, a disk label drawn independently and uniformly from the set $\{1,2,\ldots,n\}$, where $n$ is the block length of the $(n,k)$ MDS erasure code being used.
The amount of time taken to repair a disk after the $i$th failure instant $X_i$ is denoted $Z_{i+1}$. The process 
$\{Z_i,~i=2,3,\ldots\}$ is assumed to be i.i.d with pdf $f_Z(\cdot)$ and CDF $F_Z(\cdot)$. 
We also define the indicator  random variable $B_1=1$, and $B_i=1$ if $Y_i <  Z_i$, $B_i=0$,  otherwise, for $i=2,3,\ldots$.  We will use the notation $\pmb{B}_{i:j}=(B_i,B_{i+1},\ldots,B_{j})$. Our calculation is based on runs of ones ('1-runs') in $B_{2:s}$. Let $\pmb{U}_s=(U_1,U_2,\ldots,U_R)$ denote the vector of runs of 1's in $\pmb{B}_{2:s+1}$ for $s \geq 1$, where $R$ denotes the number of 1-runs in $\pmb{B}_{2:s}$.  Thus, for the sequence $\pmb{b}_{1:11}=10001111011$, $u_s=\phi$ and $r=0$, for $s=1,2,3$, $u_s=(1)$ and $r=1$, for $s=4$  and $\pmb{u}_{s}=(4,2)$, $r=2$, for $s=10$.  

The probability of data loss is given by
\begin{eqnarray}
\lefteqn{P(\mathcal{D}_t)} \nonumber  \\
& = & \sum_{s=n-k+1}^\infty \sum_{\pmb{u}}P(\mathcal{D}_t| N_t=s, \pmb{U}_s=\pmb{u})P(N_t=s, \pmb{U}_s=\pmb{u}). \nonumber \\
\label{eqn-gen-1}
\end{eqnarray}
In (\ref{eqn-gen-1}), the term $P(\mathcal{D}_t| N_t=s, \pmb{U}_s=\pmb{u})$,  is the conditional probability of the event that in one of the 1-runs of a failure vector $\pmb{B}_{2:s+1}$ with $1$-run vector $\pmb{u}$, the number of distinct disk failures exceeds $(n-k)$.    The probability that exactly $l$ distinct disks fail during a run of length $u$, denoted $\Pi_n(u+1,l)$,   is given by
\begin{equation}
\Pi_n(u+1,l)=\frac{1}{n^{u+1}} {n \choose l} \sum_{a_1>0,a_2>0,\ldots,a_l>0} {u+1 \choose a_1,a_2,\ldots,a_l},
\end{equation}
and $\Pi_n(u+1,l)=0$ for $l>u+1$.
In terms of $\Pi_n(u,l)$ we obtain 
\begin{equation}
P(\mathcal{D}_t| N_t=s,  \pmb{U}_s=\pmb{u})=1-\prod_{j=1}^r \sum_{i=1}^{n-k}\Pi_n(u_j+1,i),
\end{equation}
where $r$ is the number of 1-runs in the sequence of failures.
Note that $$\sum_{a_1>0,a_2>0,\ldots,a_l>0} {u \choose a_1,a_2,\ldots,a_l}=l! S_2(u,l),$$ where $S_2(u,l)$ is the Stirling number of the second kind.

The second term in (\ref{eqn-gen-1}), $P(N_t=s,  \pmb{U}_s=\pmb{u})$ is given by
\begin{eqnarray}
& &P(N_t=s,  \pmb{U}_s=\pmb{u}) = P(N_t=s|\pmb{U}_s=\pmb{u})P(\pmb{U}_s=\pmb{u}) \nonumber \\
& & = P(N_t=s|\pmb{U}_s=\pmb{u})N(s,\pmb{u})(1-G)^{s-w(\pmb{u})} G^{w(\pmb{u})} ,
\label{eqn:JointProb}
\end{eqnarray}
where $w(\pmb{u})=\sum_{i=1}^r u_i$, $N(s,\pmb{u})$ is the number of binary sequences $\pmb{b}_{2:s+1}$ (of length $s$) with a 1-run vector $\pmb{u}$, $G:=\int_{0}^{\infty}F_y(z)f_Z(z)dz$ is the probability that $Z_i > Y_i$ and $P(\pmb{U}_s=\pmb{u})$ denotes the probability of selecting a vector $\pmb{B}_{2:s+1}$ with 1-run vector $\pmb{u}$.

We thus obtain the following expression for the data loss probability,
\begin{eqnarray}
\lefteqn{P(\mathcal{D}_t) =  } & \nonumber \\
&=& \sum_{s=n-k+1}^{\infty}  \sum_{\pmb{u}}P(\mathcal{D}_t| N_t=s, \pmb{U}_s=\pmb{u}) \nonumber  \\ 
& & \times N(s,\pmb{u})(1-G)^{s-w(\pmb{u})} G^{w(\pmb{u})} Pr(N_t=s|\pmb{U}_s=\pmb{u})\nonumber \\
& =&\sum_{s=n-k+1}^\infty  
 \sum_{\pmb{u}} \left(1-\prod_{j=1}^r \sum_{i=1}^{n-k}\Pi_n(u_j+1,i) \right) \nonumber  \\ 
 & & \times N(s,\pmb{u})(1-G)^{s-w(\pmb{u})} G^{w(\pmb{u})} Pr(N_t=s|\pmb{U}_s=\pmb{u}).\nonumber \\
\label{eqn-dlp-2}
\end{eqnarray}

\subsubsection{A Lower Bound}
A lower bound is obtained by  writing
\begin{eqnarray}
P(\mathcal{D}_t) 
& = &  \sum_{s=n-k+1}^{\infty}\sum_{\pmb{b}}P(\mathcal{D}_t| N_t=s, \pmb{B}_{2:s+1}=\pmb{b}) \nonumber \\ 
& & \times P(N_t=s, \pmb{B}_{2:s+1}=\pmb{b}) \nonumber \\
\label{eqn-rewrite-errprob}
\end{eqnarray}
and restricting the sum in (\ref{eqn-rewrite-errprob}) to $\pmb{b} \in \mathcal{B}$ where 
\begin{eqnarray}
\mathcal{B} & := & \{1^{n-k}0^{s-(n-k)},01^{n-k}0^{s-1-(n-k)}, \nonumber \\
& & \ldots,0^{s-1-(n-k)}1^{n-k}0\}.
\end{eqnarray}
Observe that the cardinality  $|\mathcal{B}|=s-(n-k)$. Also
$$P(\mathcal{D}_t| N_t=s, \pmb{B}_{2:s+1}=\pmb{b})=\Pi_n(n-k+1,n-k+1)$$ for $\pmb{b} \in \mathcal{B}$ and $$\Pi_n(n-k+1,n-k+1)=\frac{(n-1)!}{n^{n-k}(k-1)! }.$$  
Thus we obtain 
\begin{eqnarray}
\lefteqn{P(\mathcal{D}_t) } \nonumber \\
& \geq   &  \frac{(n-1)!}{n^{n-k}(k-1)! } \sum_{s=n-k+1}^\infty \sum_{\pmb{b} \in \mathcal{B}}P(N_t=s,\pmb{B}_{2:s+1}=\pmb{b}) \nonumber \\
& = & \frac{(n-1)!}{n^{n-k}(k-1)! } \sum_{s=n-k+1}^\infty \sum_{\pmb{b} \in \mathcal{B}}P(N_t=s|\pmb{B}_{2:s+1}=\pmb{b}) \nonumber \\
& & \times P(\pmb{B}_{2:s+1}=\pmb{b}) \nonumber \\
& = & \frac{(n-1)!}{n^{n-k}(k-1)! } \sum_{s=n-k+1}^\infty \sum_{\pmb{b} \in \mathcal{B}}P(N_t=s|\pmb{B}_{2:s+1}=\pmb{b}) \nonumber \\
& & \times (1-G)^{s-(n-k)}G^{n-k} \nonumber \\
& = & \frac{G^{n-k} (n-1)!}{n^{n-k}(k-1)! } \sum_{s=n-k+1}^\infty (1-G)^{s-(n-k)}  \nonumber \\
& & \sum_{\pmb{b} \in \mathcal{B}}P(N_t=s|\pmb{B}_{2:s+1}=\pmb{b}) \nonumber \\
& \stackrel{(a)}{=} & \frac{G^{n-k} (n-1)!}{n^{n-k}(k-1)! } \sum_{s=n-k+1}^\infty (1-G)^{s-(n-k)} \nonumber \\
& & \times (s-(n-k))P(N_t=s|\pmb{B}_{2:s+1}=1^{n-k}0^{s-(n-k)}) \nonumber \\
& \stackrel{(b)}{=} & \frac{G^{n-k} (n-1)!}{n^{n-k}(k-1)! } \sum_{s=n-k+1}^\infty (1-G)^{s-(n-k)} \nonumber \\
& & \times  (s-(n-k))P(N_t=s|\pmb{B}_{2:*}=1^{n-k}0^*) 
\label{eqn:LongLowerBound}
\end{eqnarray}
where  in (a) we have used the facts (i) $|\mathcal{B}|=s-(n-k)$, and  (ii)  $P(N_t=s|\pmb{B}_{2:s+1}=\pmb{b})$ is a constant for $\pmb{b} \in \mathcal{B}$.  Observe that in (b) the conditioning event is that the infinitely long sequence $\pmb{B}_{2:*}=1^{n-k}000... =:1^{n-k}0^*$. Henceforth, we denote $\pmb{B}_{2:*}$ by $\pmb{B}$.

\subsubsection{Limiting Behavior as $G\to 0, t \to \infty$.}
The lower bound in (\ref{eqn:LongLowerBound}) is tight in the limit as $G \to 0$ since it accounts for all terms except ones that are  $o(G^{n-k})$. Thus
\begin{eqnarray}
\lefteqn{\lim_{G \to 0} \frac{  P({\mathcal D}_t) }{(G/n)^{(n-k)}} \frac{(k-1)!}{(n-1)!} } \nonumber \\  
                                                                                      & = &  \sum_{s=n-k+1}^\infty (s-(n-k))P(N_t=s|\pmb{B}=1^{n-k}0^*) \nonumber \\
                                                                                      & =  & \sum_{i=n-k+1}^\infty P(N_t \geq i | \pmb{B}=1^{n-k}0^*).
\label{eqn-limit}
\end{eqnarray}
In order to estimate (\ref{eqn-limit}), we follow the approach taken in \cite{gallager1995discrete}, Ch. 3. Define indicator random variable $I_i,~i=1,2,\ldots$
\begin{eqnarray}
I_i:=\left\{ \begin{array}{ll}
1, & i=1  \mbox{ or }  (i>1 \mbox{ and } \sum_{j=1}^{i-1}Y_j < t), \nonumber \\
0, & \mbox{ otherwise. }
\end{array}
\right.
\end{eqnarray}
Note that the events $\{I_n=1\}$ and $\{N_t \geq n-1\}$ are identical.
Now consider 
\begin{eqnarray}
\lefteqn{E\left(Y_1+Y_2+\ldots+Y_{N_t+1} | \pmb{B}=1^{n-k}0^*\right)} \nonumber\\ 
& = & E(\sum_{i=1}^\infty Y_i I_i| \pmb{B}=1^{n-k}0^*) \nonumber \\
& = & E(Y_1I_1)+E(\sum_{i=2}^{n-k+1} Y_i I_i| \pmb{B}=1^{n-k}0^*)  \nonumber \\
& & + E(\sum_{i=n-k+2}^\infty Y_i I_i| \pmb{B}=1^{n-k}0^*)\nonumber \\
&\stackrel{(a)}{=} & \underbrace{E(Y) + E(Y|Y<Z)\sum_{i=1}^{n-k}P(N_t \geq i | \pmb{B}=1^{n-k}0^*)}_{A}  \nonumber \\
& & +E(Y|Y\geq Z) \sum_{i=n-k+1}^\infty P(N_t \geq i|\pmb{B}=1^{n-k}0^*), \nonumber \\
\label{eqn-wald}
\end{eqnarray} 
where in (a) we have used the independence of $Y_i$ and $I_i$.
The left hand side in (\ref{eqn-wald}) is larger than $t$ by the mean residual time $\Delta:=E(X_{N_t + 1}-t|\pmb{B})$. Thus
\begin{eqnarray}
{\sum_{i=n-k+1}^\infty P(N_t \geq i|\pmb{B}=1^{n-k}0^*)} 
& = &  \frac{t +\Delta-A}{E(Y|Y \geq Z)}. \nonumber \\
\end{eqnarray}
Upon assuming that $E(Y^2|Y>Z)$ is finite it follows that $E(\Delta|\pmb{B})$ is finite.  Thus for large $t$
\begin{eqnarray}
\lim_{t \to \infty}\frac{1}{t}\sum_{i=n-k+1}^\infty P(N_t \geq i|\pmb{B}=1^{n-k}0^*) = \frac{1}{E(Y|Y \geq Z)}.
\end{eqnarray}
Since $E(Y|Y\geq Z) \to E(Y)$ as $G \to 0$ we obtain 
\begin{equation}
\boxed{\lim_{\stackrel{G \to 0, t \to \infty}{ tG^{n-k} \to 0}} \frac{  P({\mathcal D}_t)/t }{(G/n)^{(n-k)} } = \frac{(n-1)!}{(k-1)!} \frac{1}{E(Y)}.}
\label{eq:LimitingForm}
\end{equation}

Equation \eqref{eq:LimitingForm} leads to an \emph{interesting and useful} qualitative conclusion about worst case repair duration distributions for the  probability of data loss.  If $G$ is sufficiently small, the approximation $  P({\mathcal D}_t) \approx \frac{(n-1)!}{(k-1)!} \frac{t}{E(Y)}
 {(G/n)^{(n-k)} }$ becomes sharp. If, in addition, we assume that $F_Y(y)$ is \textit{convex} (which holds for practical failure distributions such as exponential and a subset of Weibull distributions), by Jensen's inequality we have $G = P(Y \leq Z) \leq P(Y \leq E[Z])$, and thus
\begin{equation}  P({\mathcal D}_t) \lessapprox \frac{(n-1)!}{(k-1)!} \frac{t}{E(Y)} (P(Y \leq E[Z])/n)^{(n-k)}.
\end{equation}
This means that in the limiting regime, the highest probability of data loss over a large class of failure distributions is when the time of repair is \textit{constant}. 
\section{Probability of Data Loss Conditioned on the Number of Failures for Constant Repair Duration}
\label{sec:constant}
We now turn our attention to the probability of data loss conditioned on the number of failures of each disk. The calculations hereafter consider the special yet important case of exponentially distributed failure durations and constant repair duration. We give the reader a glimpse of the main results for the case of a $(2,1)$ erasure correcting code. Analyses for general $(n,k)$ codes are presented in the subsequent subsection.
\subsection{Motivating Example: $(2,1)$ code}
\label{sec:motivating}

Suppose that we have one symbol of information stored in two disks. Let $m_1$ and $m_2$ denote the number of failures of disks $1$ and $2$, respectively. Let $X_{11},\ldots,X_{1m_1}$ and $X_{21},\ldots,X_{2m_2}$ be their random failure instants. By analyzing the failure timeline of both disks, we see that an error event occurs if and only if, for some failure instant $X_{1i}$ of disk $1$ and $X_{2j}$ of disk $2$, we have $|X_{1i} - X_{2j}| \leq t_\rep$.  Alternatively, there is no data loss if the random vector ${{\bm{X}} = (X_{11}, \ldots, X_{1m_1}, X_{21}, \ldots, X_{2m_2})}$ lies in the region:
$$\mathcal{R}^c = \left\{ \bm{x}: 0 \leq x_{1i}, x_{2j} \leq t,  |x_{1i} - x_{2j}| > t_{\rep}, \forall i,j \right\}.$$
The probability that $\bm{X} \in \mathcal{R}^c$ an be calculated exactly, as outlined next. 
Consider the permutation $\bm{\pi}$ on the set $\{1,2,\ldots,s\}$, $s=m_{1}+m_{2}$, which sorts $\bm{X}$ in ascending order. A corresponding failure pattern $\bm{f}$ is obtained by applying $\bm{\pi}$ to the vector $(1^{m_{1}}2^{m_{2}})$. Given a permutation $\bm{\pi}$, a  \emph{transition} $(i,i+1)$ is defined as a pair of consecutive positions of the failure pattern for which $f_{i} \neq f_{i+1}$, i.e. a transition identifies consecutive failure instants that correspond to distinct disks.  Let $\xi(\bm\pi)$ denote the number of transitions for a given permutation $\bm\pi$.\begin{prop} The probability of data loss $P_{\bm{m}}(\mathcal{D}_t)$ of a $(2,1)$-code given $\bm{m}=(m_{1},m_{2})$, $m_{i}$ the number of failures for disk $i$, is given by
\begin{equation}
P_{\bm{m}}(\mathcal{D}_t) = 1 - \sum_{j=1}^{s-1}  (1- j t_\rep/t)^{s} Pr(\xi(\bm\pi)=j),
\end{equation}
where $s=m_{1}+m_{2}$.
\end{prop}

\begin{proof}
\begin{eqnarray}
P_{\bm{m}}(\mathcal{D}_t^{c}) & = & \sum_{\bm\pi} P_{\bm{m}}(\mathcal{D}_t^{c}| \bm\pi) Pr(\bm\pi) \nonumber\\
&=&  \sum_{j}\sum_{\bm\pi~:\xi(\bm\pi)=j} P_{\bm{m}}(\mathcal{D}_t^{c}| \bm\pi) Pr(\bm\pi) \nonumber \\
& \stackrel{(a)}{=} & \sum_{j}\sum_{\bm\pi~:\xi(\bm\pi)=j} (1-j \trep/t)^{s} Pr(\bm\pi) \nonumber \\ &=&  \sum_{j} (1-j \trep/t)^{s} Pr(\xi(\bm\pi)=j).
\end{eqnarray}
In (a) we have used the fact that $P_{\bm{m}}(\mathcal{D}_t^{c}| \bm\pi)=\vol {\mathcal R}_{\bm{\pi}}^{c}$, where 
\begin{equation}
\begin{split}
{\mathcal R}^{c}_{\bm{\bm\pi}} = & \left\{(x_{1},x_{2},\ldots,x_{s})~:~0 \leq x_{1} \leq x_{2} \leq \ldots \leq x_{s} \leq t, \right. \\ & \quad \left.  x_{m+1}-x_{m} > \trep  \mbox{ for every transition } (m,m+1) \right\},
\end{split}
\end{equation}
and the fact that $\vol {\mathcal R}_{\bm{\pi}}^{c}=(1-j \trep/t)^{s}$, when $\xi(\bm\pi)=j$ as will be shown in Sec~\ref{sec:ErrorRegions} . 
\end{proof}
The following corollary provides the asymptotic behavior when $\trep/t$ is small.
\begin{cor} $\displaystyle \lim_{t_\rep/t \to 0} \frac{P_{\bm{m}}(\mathcal{D}_t)}{t_{\rep}/t} = {2m_1m_2}$.
\label{cor:asymptotic21}
\end{cor}
\begin{proof} We have 
\begin{equation*}\begin{split}\displaystyle & \lim_{t_\rep/t \to 0} \frac{P_{\bm{m}}(\mathcal{D}_t)}{t_{\rep}/t} = \\
& = \lim_{t_\rep/t \to 0} \sum_{j=1}^{s-1} \sum_{i=1}^{s} {s \choose i} \frac{(-1)^{i+1} j^i(t_\rep/t)^i Pr(\xi(\pi)=j)}{t_\rep /t} \\
& = {s}\sum_{j=1}^{s-1} j Pr(\xi(\pi)=j).
\end{split}
\end{equation*}
The summation in the last term---the average number of transitions in a permutation---is shown to be equal to $2m_1 m_2/s$ in  Thm.~\ref{thm:AveragesUgly} in Sec.~\ref{sec:closedForm}.
\end{proof}


\subsection{The Reliability of $(n,k)$ MDS Codes: Direct Approach}
\label{sec:closedForm}
To state the probability of data loss of an $(n,k)$ code we need some initial definitions.
Let $X_{i1},\ldots,X_{im_i}$ be the random failure instants of disk $i$ and let
\begin{equation}
\bm{X}=(X_{11},\ldots,X_{1m_{1}},X_{21},\ldots,X_{2m_{2}},\ldots,X_{n1},\ldots,X_{nm_{n}}).
\label{eq:bigVector}
\end{equation}
Denote the total number of disk failures in $[0,t]$ by  $s:=\sum_{i=1}^nm_i$.  Given a sample $\bm{x}$ drawn from the distribution of $\bm{X}$, we define the \textit{failure pattern} $\bm{f} $ as the vector obtained by applying the permutation which sorts $\bm{x}$ in ascending order to  $(1^{m_{1}}2^{m_{2}}\ldots n^{m_{n}})$ (the ties are broken arbitrarily and associated to events with zero probability). Note that the number of possible orderings of $\bm{x}$, $s!$, is the number of possible failure patterns $\bm{f}$ times $m_1!\ldots m_n!$. For example, the failure pattern for Fig.~\ref{fig:ErrorEventParallel} would be $(1,2,3,4,1,2,1)$.

%

Let $b \geq a \geq 1$ be integers. We denote by $[a,b]_{\mathbb{N}}$ the integer interval $\{i~\in \mathbb{N}:~a \leq i \leq b\}$, define its length to be $b-a$, and make the following definitions:

\begin{definition} {\bf Cluster $[a,b]_{\mathbb{N}}$}: An interval $[a,b]_{\mathbb{N}}$ such that $\left\{ f(i), a \leq i \leq b\right\}$ contains exactly $n-k+1$ distinct entries. The {\bf length} of a cluster is the length of the interval $[a,b]_{\mathbb{N}}$.\end{definition}

\begin{definition} {\bf Tight Cluster}: A cluster that does not contain a cluster of shorter length. \label{def:tightCluster}\end{definition}
\noindent Note that a \textit{transition} (in the sense of Section \ref{sec:motivating}) corresponds to a tight cluster  for a $(2,1)$ code, which by definition is of length $2$. 
\begin{definition} {\bf Minimal Cluster}: A cluster of length $n-k$. \end{definition}
\noindent
A minimal cluster is tight, but not every tight cluster is minimal. Furthermore, a cluster $[a,b]_{\mathbb{N}}$ is \emph{tight} if and only if $f(a)$ and $f(b)$ are distinct, and $\{f(i)~:~a<i<b\}$ has exactly $n-k-1$ distinct entries which are distinct from $f(a)$ and $f(b)$.
\begin{exe}\label{ex:failurePattern} Consider the failure pattern $(1,2,3,4,1,1,2,1)$ for a $(4,2)$ code. In this case, $[1,3]_{\mathbb{N}}$,$[2,4]_{\mathbb{N}}$, $[3,5]_{\mathbb{N}}$, $[3,6]_{\mathbb{N}}$, $[4,7]_{\mathbb{N}}$ and $[4,8]_{\mathbb{N}}$ are clusters. All but $[3,6]_{\mathbb{N}}$ and $[4,8]_{\mathbb{N}}$ are tight clusters, while $[1,3]_{\mathbb{N}}$,$[2,4]_{\mathbb{N}}$, and $[3,5]_{\mathbb{N}}$ are minimal clusters.
\end{exe}

Tight clusters correspond to critical sucessive failures that may cause data loss.
\begin{definition}
Let $\bm{b} = (b_1,\ldots, b_{l}), l \leq s-1,$ be a binary vector. The {\bf restriction} of $\bm{b}$ to an interval $[u,v]_{\mathbb{N}}$ is $\bm{b}([u,v]_{\mathbb{N}})=(b_u,\ldots,b_{v-1})$.
\end{definition}
\begin{definition}\label{def:errorRegion} {\bf Region associated with $\bm{b}$}
$$\mathcal{R}_{\bm{b}} = \left\{ (x_1,\ldots,x_s): \begin{array}{c} 0 \leq x_1 \leq \ldots \leq x_s \leq t \\
x_{i+1} - x_i < t_{\rep} \mbox{ if } b_i = 1 \\
x_{i+1} - x_i \geq t_{\rep} \mbox{ if } b_i = 0\end{array} \right\}.$$
\end{definition}
\begin{rmk}
Often some of the successive differences are unconstrained. For example, if $x_{2}-x_{1}> \trep$, and $x_{5}-x_{4} < \trep$ and $s=6$, then $\bm{b}$ should be written as $0**1*$, where $*$ in position $i$ indicates that no constraint is imposed between $x_{i+1}$ and $x_{i}$. As we will see later, as far as volume calculations are concerned nothing is lost  by considering $\bm{b}$ to be $01$, i.e. omitting the $*$'s and
writing $\bm{b}$ as $0^{i}1^{j}$ where $i$ is the number of $\geq$ constraints and $j$ is the number of $<$ constraints.
\end{rmk}
\begin{definition}{\bf Fundamental Simplex $\mathcal S$}: $\{\bm{x}: 0\leq x_1\leq x_2\leq \ldots \leq x_{s}\leq t \}$.
\end{definition}
\begin{definition}{\bf Volume Polynomial}. Given a subregion $\hat{\mathcal S}$ of the fundamental simplex $\mathcal S$, we define volume polynomial $v(\rho)=(s!/\trep^{n} )\vol \hat{\mathcal S} $, where $\rho:=t/\trep$. If $\hat{\mathcal S}=\mathcal{R}_{\bm{b}}$, then
we will use the notation $v_{\bm{b}}(\rho)$. As will be seen later, the volume polynomial depends on $\bm{b}$ through the number of constraints. Thus if $\bm{b}$ contains $i$ zeros and $j$ ones, corresponding to  $i+j$ constraints, we will write $v_{ij}(\rho)$ interchangeably with $v_{\bm{b}}(\rho)$.
\end{definition}
\subsection{Characterization of data loss event}
\label{sec:characterization}
When there are no consecutive repeated elements in $\bm{f}$, we consider that data loss occurs if there is an ordered sequence of failures $x_i,\ldots,x_{i+n-k+1}$ from $n-k+1$ different disks such that $x_{j+1}-x_j < t_\rep$, for all $j = i, \ldots i+n-k+1$. When there is at least one repeated number in the failure pattern (for example $\bm{f} = (1,1,2,2,4,4,3)$) we assume that there is a data loss event if there exists an ordered sequence $(x_{i},\ldots,x_{i+l})$ from more than $n-k+1$ disks such that $x_{j+1}-x_j < t_\rep$. 

We have two equivalent characterizations of an error event, given a failure pattern $\bm{f}$:
\begin{itemize}
\item[(i)] A binary vector $\bm{b}$ is a {\bf no-error vector} if the restriction of $\bm{b}$ to \emph{every} tight cluster  of $\bm{f}$ has weight at most $(l-1)$, where $l$ is the length of that tight cluster. 
\item[(ii)] The vector $\bm{b}$ is an {\bf error vector} if its restriction to \emph{at least} one tight cluster of length $l$ has weight $l$. 
\end{itemize}
Let us call $B_{\bm{f}}$ the set of all error vectors $\bm{b}$ for a given failure pattern $\bm{f}$. 
\begin{exe}
Consider a $(4,2)$ MDS code with  $\bm{m}=(2,2,1,1)$ and suppose the failure pattern is $121234$. Then ${\mathcal B}_f$ consists of the error vectors $**110$, $**011$ and $**111$. Following our convention of dropping the $*$'s and writing $\bm{b}$ as $0^i1^j$ we write ${\mathcal B}_f=\{2(0^11^2),0^01^3 \}$.
\end{exe}
From simple observations, one can find the following expression for $P_{\bm{m}}(\mathcal{D}_t)$.
\begin{thm} The probability of data loss satisfies
\begin{equation}P_{\bm{m}}(\mathcal{D}_t) =
\frac{1}{\rho^{s}{s \choose m_{1},m_{2},\ldots,m_{s}}}\sum_{\bm{f}} \sum_{\bm{b} \in B_{\bm{f}}}v_{\bm{b}}(\rho).
\label{eqn:thmMain}
\end{equation}
\label{thm:thmMain}
\end{thm}
\begin{proof} Let $\bm{\hat{X}}$ be the random vector associated to the ordered failure times. 
Let $P(\bm{f})$ be the probability that $\hat{\bm{X}}$ has pattern $\bm{f}$.
\begin{equation*}
\begin{split}
P_{\bm{m}}(\mathcal{D}_t) &= \sum_{\bm{f}} P_{\bm{m}}(\mathcal{D}_t | \bm{f}) P(\bm{f}) \\ &= {s \choose m_1,\ldots,m_n}^{-1}\sum_{\bm{f}} P_{\bm{m}}(\mathcal{D}_t | \bm{f}) \\ &\stackrel{(a)}{=} {s \choose m_1,\ldots,m_n}^{-1}\sum_{\bm{f}} \sum_{\bm{b} \in B_{\bm{f}}} P_{\bm{m}}(\bm{\hat{X}} \in \mathcal{R}_{\bm{b}}) \\ & \stackrel{(b)}{=} \frac{m_1!m_2! \ldots m_n!}{t^s} \sum_{\bm{f}} \sum_{\bm{b} \in B_{\bm{f}}} \vol{\mathcal{R}_b}
\end{split}
\end{equation*}
where (a) is due to the characterization of a data loss event, given $\bm{f}$, and (b) follows from the fact that the set of ordered vectors $\hat{X}$ has volume $t^s/s!$.
\end{proof}

Thus, to give explicit forms for $P_{\bm{m}}(\mathcal{D}_t)$, we need two elements
\begin{itemize}
\item[(i)] Computations of the volume of the error regions $\mathcal{R}_b$, or equivalently, computation of the volume polynomial $v_{\bm{b}}(\rho)$.
\item[(ii)] Enumeration of the set of error vectors $B_{\bm{f}}$.
\end{itemize}
The volume computation is addressed in Sec.~\ref{sec:ErrorRegions}. We address the problem of enumerating the error vectors in this section and use Thm.~\ref{thm:volPol}, Sec.~\ref{sec:ErrorRegions} in order obtain the asymptotic behavior of $P_{m}(\mathcal{D}_t)$ as $t_{\rep}/t \to 0$.

Thm. \ref{thm:volPol}, Sec.~\ref{sec:ErrorRegions} gives a formula for computing $v_{\bm{b}}(\rho)$.
In particular, it shows that if  $\bm{b}$ is some permutation of $0^i1^j$, i.e. $w(\bm{b}) = j$, then
\begin{equation}
v_{ij}(\rho) = \frac{s!}{(s-j)!} \rho^{s-j} + O(\rho^{s-j-1}),
\end{equation}
where $s$ is the number of failures in $[0,t]$.
This means that the dominant terms in $P_{\bm{m}}(\mathcal{D}_t)$ are when $w(\bm{b}) = n-k$. In this case
\begin{equation}
v_{i,n-k}(\rho)=\frac{s!}{(s-(n-k))!}\rho^{s-(n-k)}+ O(\rho^{s-(n-k+1)}).
\end{equation}
Note also that dominant terms correspond to minimal failure clusters (i.e., of length $(n-k)$). This characterization suffices to prove the asymptotic behavior of $P_{m}(\mathcal{D}_t)$ as $t_{\rep}/t \to 0$. Let $j_{\bm{f},n-k}$ be the number of minimal failure clusters in $\bm{f}$. We have
\begin{equation}
\begin{split}
P_{\bm{m}}(D_t) &= \frac{s!}{(s-(n-k))!}\rho^{-(n-k)} \sum_{\bm{f}} \frac{j_{\bm{f},n-k}}{{s \choose m_1,m_2,\ldots,m_n}}+\\
&+O(\rho^{-(n-k+1)})
\end{split}
\end{equation}
Thus
\begin{equation}
\lim_{\rho \to \infty} {P_{\bm{m}}(D_t)}{\rho^{n-k}} = \frac{s!}{(s-(n-k))!} \sum_{\bm{f}} \frac{j_{\bm{f},n-k}}{{s \choose m_1,m_2,\ldots,m_n}}.
\label{eq:asymptoticsPDt}
\end{equation}
As will be shown later in Corollary \ref{cor:AveragesUgly}, 
\begin{equation}\begin{split} &A_{n-k} := \sum_{\bm{f}} \frac{j_{\bm{f}}}{{s \choose m_1,m_2,\ldots,m_n}} = \\
&= \frac{(n-k+1)! (s-(n-k))!}{s!} \sum_{(i_1,\ldots,i_{n-k+1}) \above 0pt \mbox{{\tiny distinct}}}{m_{i_1}\ldots m_{i_{n-k+1}}},
\end{split}
\label{eq:formulaToProve}
\end{equation}
which leads to 
\begin{equation}
\boxed{\begin{split} &\lim_{\rho \to \infty} {P_{\bm{m}}(D_t)}{\rho^{n-k}} = \\ &(n-k+1)!   \sum_{(i_1,\ldots,i_{n-k+1}) \above 0pt \mbox{{\tiny distinct}}}{m_{i_1} m_{i_2}\ldots m_{i_{n-k+1}}}. \end{split}}
\label{eq:asymptoticsPDt-final}
\end{equation}

\begin{rmk} The contribution to $P_{\bm{m}}(D_t)$ from data loss events related to non-minimal clusters is negligible in the limit $\trep/t \to 0$. 
\label{rmk:asymptoticallyNegligible}
\end{rmk}

\begin{rmk}A result with very similar flavor was proved in \cite[Ch. 6]{VenkThesis2012}, in spite of the difference between the models.  The approximations in  \cite[Sec. 6.3.2]{VenkThesis2012} show that the dominant term in the mean time to data loss is due to a ``direct path" of failures from $n-k+1$ different disks. This is completely analogous to the fact that the dominant term in $P_{\bm{m}}(D_t)$ is due to minimal clusters (i.e., to the probability associated to a succession of failures from \textit{exactly} $n-k+1$ disks).
\end{rmk}

%

\subsection{Upper Bounding the Error Term}
By enumerating all failure patterns, we calculate $P_{\bm{m}}(\mathcal{D}_t)$ explicitly. However, combinatorial upper bounds for the error terms may be useful. We derive an asymptotically optimal bound in this subsection.

Given a failure pattern $\bm{f}$, there is an error if the restriction  of the vector $\bm{b}$ to at least one tight cluster of length $l$ has weight $l$ (see characterization (ii) at the start of Sec. \ref{sec:characterization}). Let $I_1,\ldots,I_p$ be the tight clusters of $\bm{f}$ ($I_j = [a_j,b_j]_{\mathbb{N}}$). Let $l_j$ be the length of the $j$-th tight cluster.
\begin{equation}\begin{split} P_m(\mathcal{D}_t | f) &= P\left(b(I_1) = {1}^{l_1} \mbox{ or } b(I_2) = {1}^{l_2} \mbox{ or }\ldots b(I_p) = {1}^{l_p}\right) \\ &\leq \sum_{j=1}^p P\left(b(I_j) = 1^{l_j}\right)=\frac{1}{t^s}\sum_{j=1}^p \vol R_{1^{l_j}} \\ &= \sum_{l=n-k}^s j_{\bm{f},l} \vol R_{1^{l}},\end{split}\end{equation}
where we define $ j_{\bm{f},l} $ to be the number of tight clusters of length $l$ in $\bm{f}$. From the above inequality:

\begin{equation}\begin{split} P_m(\mathcal{D}_t) & \leq \frac{m_1!\ldots m_n!}{t^s} \sum_{\bm{f}} \sum_{l=n-k}^s j_{\bm{f},l} \vol R_{1^{l}} \\ &= \frac{m_1!\ldots m_n!}{t^s}  \sum_{l=n-k}^s \sum_{\bm{f}} j_{\bm{f},l} \vol R_{1^{l}} \\ &= \frac{s!}{t^s} \sum_{l=n-k}^s \vol R_{1^{l}} \underbrace{\left(\sum_{\bm{f}} j_{\bm{f},l}/{ s \choose {m_1,\ldots,m_n}}\right)}_{:= A_{l}}. \end{split}\end{equation}
But $A_{l}$ is the average number of tight clusters of length $l$. Also note that $l = n-k$ is the dominant term. Hence this upper bound collapses with exact calculation for vanishing time of repair.

The following theorem gives a closed form expression for $A_{l}$.


\begin{thm} Let $\mathcal{I}_{n-k+1}$ be the set of all $(n-k+1)$-tuples of distinct numbers $(i_1,\ldots,i_{n-k+1})$, $1 \leq i_j \leq n$.

\begin{equation}\begin{split} A_{l} = &(s-l){ s \choose l+1}^{-1} \times \\& \sum_{(i_1,\ldots i_{n-k+1}) \in \mathcal{I}_{n-k+1}} \sum_{{q_i \geq 1} \above 0pt \sum q_i = l-1} m_{i_1} m_{i_{n-k+1}} \prod_{j=2}^{n-k-1}{ m_{i_j} \choose {q_j}}. 
\end{split}\end{equation}
\label{thm:AveragesUgly}
\end{thm}
\begin{proof}
Given a failure pattern $\bm{f}$, let $Y_j$, $j = 1, \ldots, s-l$, be indicator random variables which are $1$ if $[j,{j+l}]_{\mathbb{N}}$ is a tight cluster and $0$ otherwise. We would like to calculate $A_l = E[Y_1+\ldots+Y_{s-l}] = (s-l) E[Y_1]$. But $E[Y_1]$ is the probability that $[1,{l+1}]_{\mathbb{N}}$ is a tight cluster of $\bm{f}$. We use Definition ~\ref{def:tightCluster} (and the corresponding lemma) to calculate this probability. Pick a random pattern $(F_1,\ldots,F_{l+1})$ (there are ${s \choose l+1}$ ways of doing so). A tight cluster is formed by choosing two different numbers for endpoints $F_1$ and $F_{l+1}$ (say $i_1$ and $i_{n-k+1}$), and then choosing $(n-k-1)$ other numbers ($i_2,\ldots,i_{n-k}$) to fill the remaining $(l-1)$ positions. If $l>n-k$, some of the numbers will appear more than once in $f_2,\ldots,f_{l}$. Suppose that $i_j$ appears $q_j$ times (there are ${m_{i_j} \choose q_j}$ ways in which this happens). Since $i_{1}$ and $i_{n-k+1}$ appear only once, the total choices for the pattern are the product between $m_{1} m_{n-k+1}$  and the choices for $i_2,\ldots, i_{n-k}$. Summing over all possible $q_j$ gives us the final answer.
\end{proof}
\begin{cor} $A_{n-k}$ is given by Equation \eqref{eq:formulaToProve}\label{cor:AveragesUgly}
\end{cor} 

We now estimate the probability of data loss (or equivalently, the reliability) of an erasure coded storage system with Poisson failures.

\begin{thm} For Poisson distributed failures with rate parameter $\lambda$ and constant repair time $\trep$
\begin{equation}\displaystyle \lim_{t_\rep\to 0} \frac{P(\mathcal{D}_t)}{t_\rep^{n-k}} = \frac{n!}{(k-1)!} \lambda^{n-k+1} t
\label{eqn:asymptoticConst}
\end{equation}
\label{thm:asymptoticBehavior}
\end{thm}
\begin{proof} Let $M_i \sim \mbox{Poisson}(\lambda t)$ be the random variable associated to the number of failures until time $t$.
\begin{equation*} \begin{split} &\lim_{t_\rep\to 0} \frac{P(\mathcal{D}_t)}{t_\rep^{n-k}} = \lim_{t_\rep\to 0}\sum_{\bm{m}}\frac{P_m(\mathcal{D}_t)}{t_\rep^{n-k}} P(\bm{M}=\bm{m}) \\ & \stackrel{(a)}{=} \sum_{\bm{m}}\lim_{t_\rep\to 0}\frac{P_m(\mathcal{D}_t)}{t_\rep^{n-k}} P(\bm{M}=\bm{m}) \\&\stackrel{(b)}{=} \frac{(n-k+1)!}{t^{n-k}} \sum_{\bm{m}} \sum_{(i_1,\ldots,i_{n-k+1}) \above 0pt \mbox{{\tiny distinct}}}{m_{i_1}\ldots m_{i_{n-k+1}}} P(\bm{M} = \bm{m}) \\& = \frac{(n-k+1)!}{t^{n-k}} \sum_{(i_1,\ldots,i_{n-k+1}) \above 0pt \mbox{{\tiny distinct}}}{E[M_{i_1}] \ldots E[M_{i_{n-k+1}}]} \\
& = (n-k+1)! {n \choose n-k+1} \lambda^{n-k+1} t.
\end{split}
\end{equation*}
Interchanging the limit and summation in $(a)$ is justified by bounded convergence (since $P_{\bm{m}}(\mathcal{D}_t)/t_\rep^{n-k}$ is naturally uniformly bounded). Step $(b)$ follows from the asymptotics derived in \eqref{eq:asymptoticsPDt}.
\end{proof}

\subsection{Possible Generalizations}
The machinery developed in this section has some implications to the reliability of other failure point processes. In a fairly general setup, suppose that the failure mechanism is such that the joint probability density between the random failure instants of disks $1$ to $n$ (cf Eq. \ref{eq:bigVector}), conditioned on the number of failures, is given by $g(\bm{x})$. If $g(\bm{x})$ is bounded for all $\bm{x}$, we have the following qualitative result:

 \begin{thm} As $t_\rep \to 0$, the probability of data loss $P_{\bm{m}}(\mathcal{D}_t)$ of an erasure coded system cannot decay slower $t_\rep^{n-k}$.
\label{thm:asymptoticGeneral}
\end{thm}
\begin{proof} The characterization of a data loss event in Sec \ref{sec:characterization} does not depend on the statistics of the system. The calculations in Theorem \ref{thm:thmMain} can be thus carried out replacing $P(\hat{\bm{X}} \in \mathcal{R}_{\bm{b}})$ by 
\begin{equation} \int_{\mathcal{R}_{\bm{b}}} \hat{g}(\bm{x}) d \bm{x},
\end{equation}
where $\hat{g}(\bm{x})$ is the pdf of the order statistics obtained by sorting  $\bm{x}$ in ascending order. Since we assumed that $g(\bm{x})$ is bounded, so is $\hat{g}(\bm{x})$, and hence the above integral can be upper bounded by a constant times $\vol \mathcal{R}_{\bm{b}}$. The result now follows from the asymptotic analysis of $\vol \mathcal{R}_{\bm{b}}$, provided by Eq. \eqref{eq:asymptoticsPDt}.
\end{proof}
Notice that this result does not assume independence on the failure time or invariance under time.

\begin{exe}[Non-Homogeneous Poisson Processes]
 This type of process can model situations such as \textit{aging} effects and reliability growth. Let $\lambda(x)$ be a function of time (referred to as the \textit{rate function}). For this process, the probability that there are $m$ failures of one disk in the interval $[a, a+h]$ is given by:
\begin{equation} P(M_i(a,a+h) = m) = \frac{(\Lambda(a,a+h))^{m} \exp{\left[-\Lambda(a+h,a)\right]}}{m!},
\end{equation}
where $\Lambda(a,a+h) := \int_{a}^{a+h} \lambda(x) \mbox{d}x$. Define the \textit{normalized rate function} as 
\begin{equation}
\tilde{\lambda}(x) =  \frac{\lambda(x)}{\Lambda(0,t)}.
\end{equation}

It is not hard to see that the \textit{failure times} of a disk are independently, identically distributed with pdf $\tilde{\lambda}(x)$ (see, for example, \cite[p. 64]{Queue}, adapted to the non-homogeneous case). This way, the joint pdf of the ordered failure times, conditioned on the number of failures is $m$ is given by.
\begin{equation} g(\hat{x}_1,\ldots,\hat{x}_m) = s! \tilde{\lambda}(x_1)\tilde{\lambda}(x_2)\ldots \tilde{\lambda}(x_m),
\end{equation}

Now let $C(t) = \max_{x \in [0,t]} \tilde{\lambda}({x})$. We can bound ${P(\hat{\bm{X}} \in \mathcal{R}_{\bm{b}})}$ by $s! C(t)^s \mbox{vol } \mathcal{R}_b$, and thus

\begin{equation} P_m(\mathcal{D}_t) \leq s! C(t)^s \sum_{\bm{f}} \sum_{\bm{b} \in \bm{f}} \left( \mbox{\upshape{vol }} \mathcal{R}_b \right)(P_{\bm{m}}(\bm{f})).
\label{eq:maxBound}
\end{equation}
A special case of this process are the ``Power-Law Processes", where $\lambda(x) = \lambda x^{\beta}$, $\lambda > 0$. In this case,  $C(t) = (\beta+1)\frac{1}{t}.$

\end{exe}

\section{Set Avoidance Probabilities for Cartesian Products of Random Sets}
\label{sec:avoidance}
The closed form calculations performed in the previous sections are particularly useful to characterizing the asymptotic behavior of the system. The objective of this section is to provide a simple upper bound  based on Jensen's inequality. The proofs of the theorems, as well as an upper bound based on the inclusion-exclusion principle, along with a geometric characterization of situations when these bounds are tight, can be found in Appendix~\ref{app:setAvoidance} and in \cite{CampVai2}. The set avoidance lower bound is used to derive a lower bound on the reliability function in Sec.~\ref{sec:avoidance:app}, some examples are presented in Sec.~\ref{sec:avoidance:ex} and general results for $(n,k)$ MDS codes are presented in Sec.~\ref{sec:avoidance:MDS}.

\subsection{Lower Bounds}
\label{sec:avoidance:lb}
 Given sets $S_{1}$, $S_{2}$, $\mathcal{R} \subset S_{1} \times 
S_{2}$ and $x_{1} \in S_{1}$ we define the shadow of a section of $\mathcal R$ as $\mathcal{R}_{1}(x_{1})=\{x_{2} \in S_{2}~:~(x_{1},x_{2}) \in \mathcal{R}\}$. In the following, the operator $\times$ has precedence over set operations such $\bigcap$ and $\bigcup$.

\begin{lemma}
Let ${\mathcal X}:=\{X_{1},X_{2},\ldots,X_{m_1}\}$ and ${\mathcal Y}:=\{Y_{1},Y_{2},\ldots,Y_{m_2}\}$, where the $X_{i}$'s are i.i.d on a set ${\mathcal S}_{1}$ and the $Y_{i}$'s are i.i.d on a set ${\mathcal S}_{2}$. Let $X$, $Y$ be generic random variables distributed as $X_{i}$ and $Y_{i}$, resp. Let $\mathcal{R} \subset {\mathcal S}_{1} \times {\mathcal S}_{2}$. Then 
\begin{equation}
P\left({\mathcal X}\times {\mathcal Y} \bigcap \mathcal{R} =\emptyset\right) \geq \left( P(\{X\} \times {\mathcal Y} \bigcap \mathcal{R}=\emptyset)\right)^{m_1}
\end{equation}
and equality holds iff $P(X  \in \bigcup_{i=1}^{m_2} \mathcal{R}(y_{i}))$ is a constant for $(y_{1},y_{2},\ldots,y_{m_2})$ with positive pmf.
\label{lem:avoidance}
\end{lemma}
\begin{cor}
\begin{equation}
P\left({\mathcal X}\times {\mathcal Y} \bigcap \mathcal{R} =\emptyset\right) \geq P((X,Y) \notin  \mathcal{R})^{m_1 m_2}.
\end{equation}
\label{Cor:main}
\end{cor}

\subsection{Application of Set Avoidance Calculations to the Data Loss Probability Calculation}
\label{sec:avoidance:app}
We first apply the bounds developed in the previous section to derive a lower bound on the reliability function ${R}(t)$. The bound is given in terms of the volume of the error
region associated with a given code. A systematic method for calculating the volume of the error region is then presented along with an overview of some of the theoretical results related to the calculation of an error polynomial associated with the code.
Proofs are presented in the next session and in the appendix.

For the avoidance upper bound, we need a different definition of a data loss event. Let $\bm{f}$ be a failure pattern. We consider that data loss occurs if there is an ordered sequence of failures $x_i,\ldots,x_{i+n-k+1}$ from $n-k+1$ different disks such that $x_{j+1}-x_j < t_\rep$, for all $j = i, \ldots i+n-k+1$, even when the failure pattern has repeated consecutive elements. From Remark~\ref{rmk:asymptoticallyNegligible}, this characterization is asymptotically the same as the one in \ref{sec:characterization}.

Let $\mathcal{R}$ be the region
\begin{eqnarray}
\mathcal{R} & := & \{ (x_1,x_2,\ldots,x_n) \in [0,t]^n\ :\ |x_{i_1}-x_{i_2}|< t_{rep}, \nonumber \\ & & |x_{i_2}-x_{i_3}| < t_{rep},\ldots,|x_{i_{n-k+1}} -x_{i_{n-k}}| < t_{rep}  \nonumber \\ & &\mbox{ for some } i_1,i_2,\ldots,i_{n-k+1}\}.
\label{eqn:error-region}
\end{eqnarray}
Note that $\mathcal{R} \subset [0,t]^n$ contains the error regions of a code when there is precisely one failure of each disk ($s = m_1+\ldots+m_n=n$ and $m_i = 1$).
Suppose that in the interval $[0,t]$, disk $i$ fails $m_i>0$ times. Let $\bm{m}=(m_1,m_2,\ldots,m_n)$.  The failure instants of the $i$-th disk are denoted by ${\mathcal X}_i = \{X_{i1},\ldots,X_{i m_i}\}$, where the $X_{ij}$ are independently and uniformly drawn on the time interval $[0,t]$.
 A data loss event occurs if and only if ${\mathcal X_{1}} \times {\mathcal X_{2}} \times \ldots \times {\mathcal X_{n}}\bigcap {\mathcal R} \neq \emptyset$, where ${\mathcal R}$ is error region for the code as defined in (\ref{eqn:error-region}). Let $X_{i},~i=1,2,\ldots,n$ be i.i.d. random variables, uniformly distributed on $[0,t]$ and let $\bm{X}=(X_{1},X_{2},\ldots,X_{n})$.  The following proposition follows immediately  from Cor.~\ref{Cor:main}.

\begin{thm} The probability that there is no data loss in the interval $[0,t]$, given $m_{i}$, the number of failures for disk $i$ in $[0,t]$,  $m_i > 0$, $i=1,2,\ldots,n$ satisfies
\begin{equation} \begin{split} P_{\bm{m}}(\mathcal{D}_t^c) \geq P_{\bm{m}}( \bm{X} \in \mathcal{R}^c)^{m_1 m_2 \ldots m_n} = \left(1-\frac{\mbox{\upshape{vol }}\mathcal{R}}{t^n}\right)^{m_1 \ldots m_n}.
\end{split}
\label{eq:Bound1}
\end{equation}
\label{thm:Bound}
\end{thm}
\begin{proof} The quantity on the left is $P({\mathcal X_{1}} \times {\mathcal X_{2}} \times \ldots \times {\mathcal X_{n}}\bigcap {\mathcal R} = \emptyset)|\bm{M}=\bm{m})$. Thus the inequality in \eqref{eq:Bound1} follows directly from Cor.~\ref{Cor:main}. The equality in \eqref{eq:Bound1} follows from the fact that $X_{i}$'s are uniform random variables iid over $[0,t]$. 
\end{proof}

We proceed to calculate the volume of ${\mathcal R}$ for a few example codes, and then state a general result. 

\subsection{Graphical Representation of Constraints, Some Example Error Regions }
\label{sec:avoidance:ex}
In order to help calculate the volume of the error region $\mathcal{R}$ defined in (\ref{eqn:error-region}) we 
consider a binary vector $\bm{b}=(b_1, b_2, \ldots, b_{l}), l \leq n-1$ and to define $\mathcal{R}_{\bm{b}}\in [0,t]^n$ as the region
$$\mathcal{R}_{\bm{b}} = \left\{ (x_1,\ldots,x_n) \in [0,t]^n: \begin{array}{c} x_1 \leq \ldots \leq x_n  \\
x_{i+1} - x_i \leq t_{\rep} \mbox{ if } b_i = 1 \\
x_{i+1} - x_i > t_{\rep} \mbox{ if } b_i = 0\end{array} \right\}.$$
Note that except for the dimension of the binary vector this definition coincides with  Def. \ref{def:errorRegion}.

The  vector $\bm{b}$ is  conveniently visualized as a graph $G_{\bm{b}}$ with $n$ vertices such that there is an edge between $i$ and $i+1$ iff $b_i = 1$. The region $\mathcal{R}$ can be decomposed into a disjoint union of regions $\mathcal{R}_{\bm{b}}$, the union being over all edges $\bm{b}$ that are {\bf error vectors}.

In cases where there are no constraints between successive failure instants, the dimension of the vector is reduced and the corresponding graph has fewer nodes. As an example consider an ordered vector of failure instants $\bm{x}=(x_{1},x_{2},x_{3},x_{4})$ with the constraints $x_{2}-x_{1}> \trep$, $x_{3}-x_{2} < \trep$. This constraint is represented by the vector $\bm{b}=01$, and is shown as a graph with three vertices.

A systematic method for calculating the volume of $\mathcal{R}_{\bm{b}}$ and hence of $\mathcal{R}$ is presented in Sec.~\ref{sec:ErrorRegions}. Here we show by example, the error vectors that correspond to specific codes.

\begin{exe}
{(n,n-1)-single parity code.}
\label{sub:Simplex}
In general, if $k=n-1$, any simultaneous two disk failures (within an interval of length $t_{\text{rep}}$) will cause data loss. Therefore
\begin{equation*} \mathcal{R}= \left\{ (x_1,\ldots, x_n) \in [0,t]^n: \exists i \neq j \mbox{ s.t. } |x_i - x_j| \leq t_{\text{rep}} \right\}, \mbox{and}
\end{equation*}
\begin{equation*} \mathcal{R}^c = \left\{ (x_1,\ldots, x_{n}) \in [0,t]^n: |x_i - x_j| > t_{\text{rep}} \mbox{ for all } i,j\right\}.
\end{equation*}
Fig.~\ref{fig:regionR3} is an illustration of region $\mathcal{R}^c$ in three dimensions ($n = 3, k = 2$),
\begin{figure}[t]
\centering
\subfloat{
\includegraphics[scale=0.3]{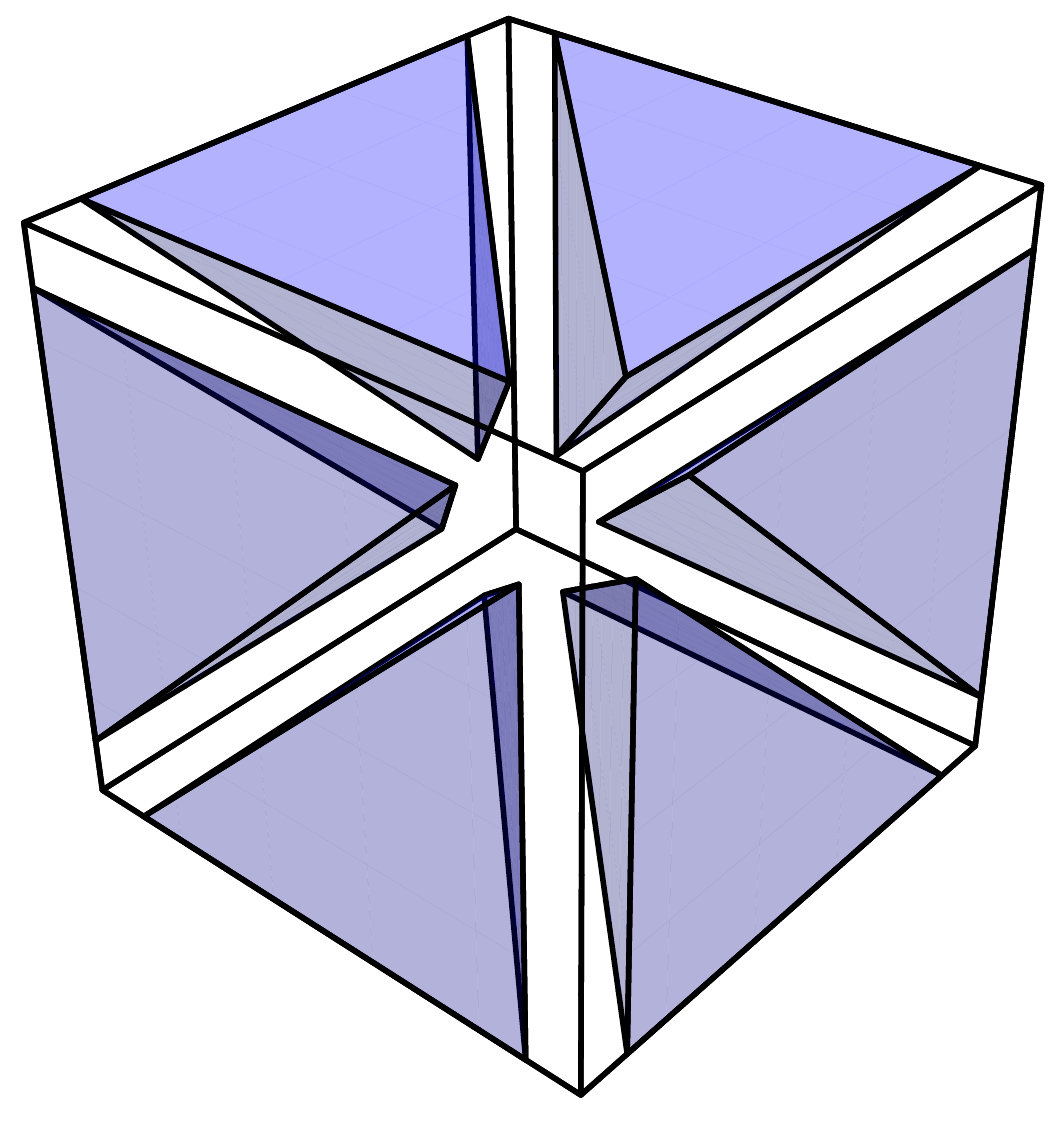}
}
\hfill\subfloat{

\includegraphics[scale=0.3]{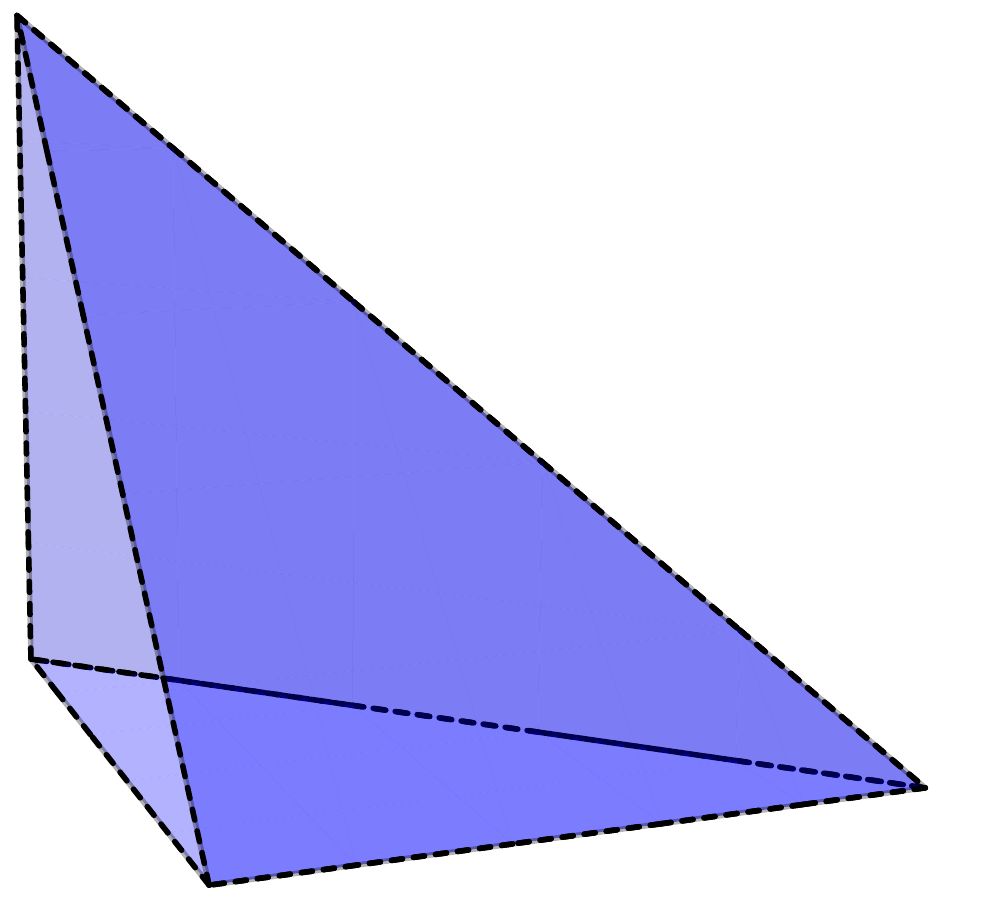} }
\\
\subfloat{\includegraphics[scale=0.5]{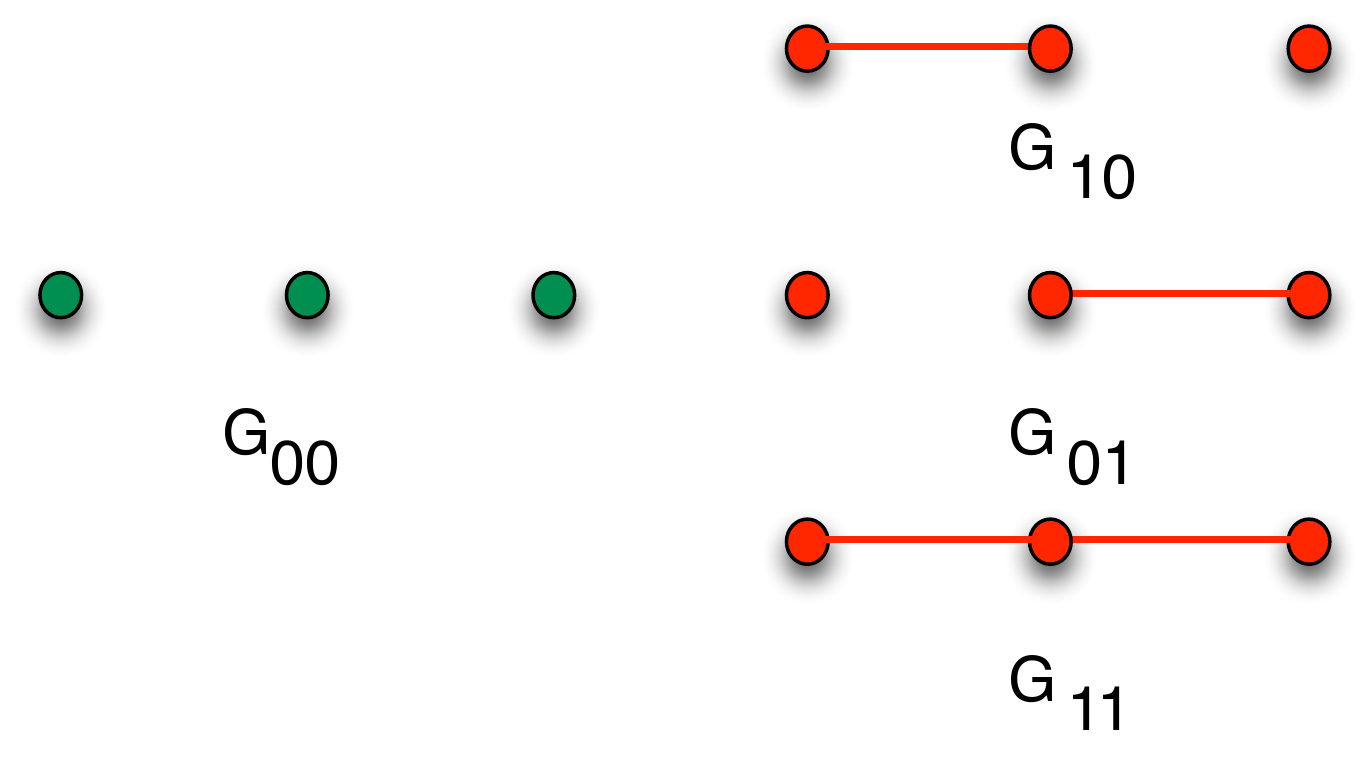}}
\caption{Illustration of the error region $\mathcal{R}^c$ and no-error region $\mathcal{R}^c$ for the $(3,2)$ code. Also shown is a single simplex, corresponding to one of the orderings of $(x_1,x_2,x_3)$. The error region is the `star-shaped' region that is unshaded. The error region is the disjoint union of the polytopes $\mathcal{R}_{\bm{b}}$. The corresponding error graphs  are  $G_{\bm{b}}$, with error vectors $10$, $01$ and $11$.}
\label{fig:regionR3}
\end{figure}
The fact that the above region is a simplex is proved in Appendix A, Lemma \ref{lemma:Auxiliar}. It is also proved that
\begin{equation*}
\mbox{vol }\mathcal{R}^c =(t-(n-1)t_{\text{rep}})^n.
\end{equation*}
Also shown in Fig. \ref{fig:regionR3} is a graphical representation for the error and non-error vectors.
\end{exe}
\begin{rmk}
For the analysis above, we require that ${t \geq (n-1) t_{\text{rep}}}$.  
\end{rmk}

For general codes the no-error regions are not elementary simplices as in a $(n,n-1)$-code. However, a systematic method
for calculating the volume of an error region is presented in Sec.~\ref{sec:ErrorRegions}. 

\begin{exe}{$(4,2)$-Code: }
\begin{figure}[t]
\centering
\includegraphics[scale=0.5]{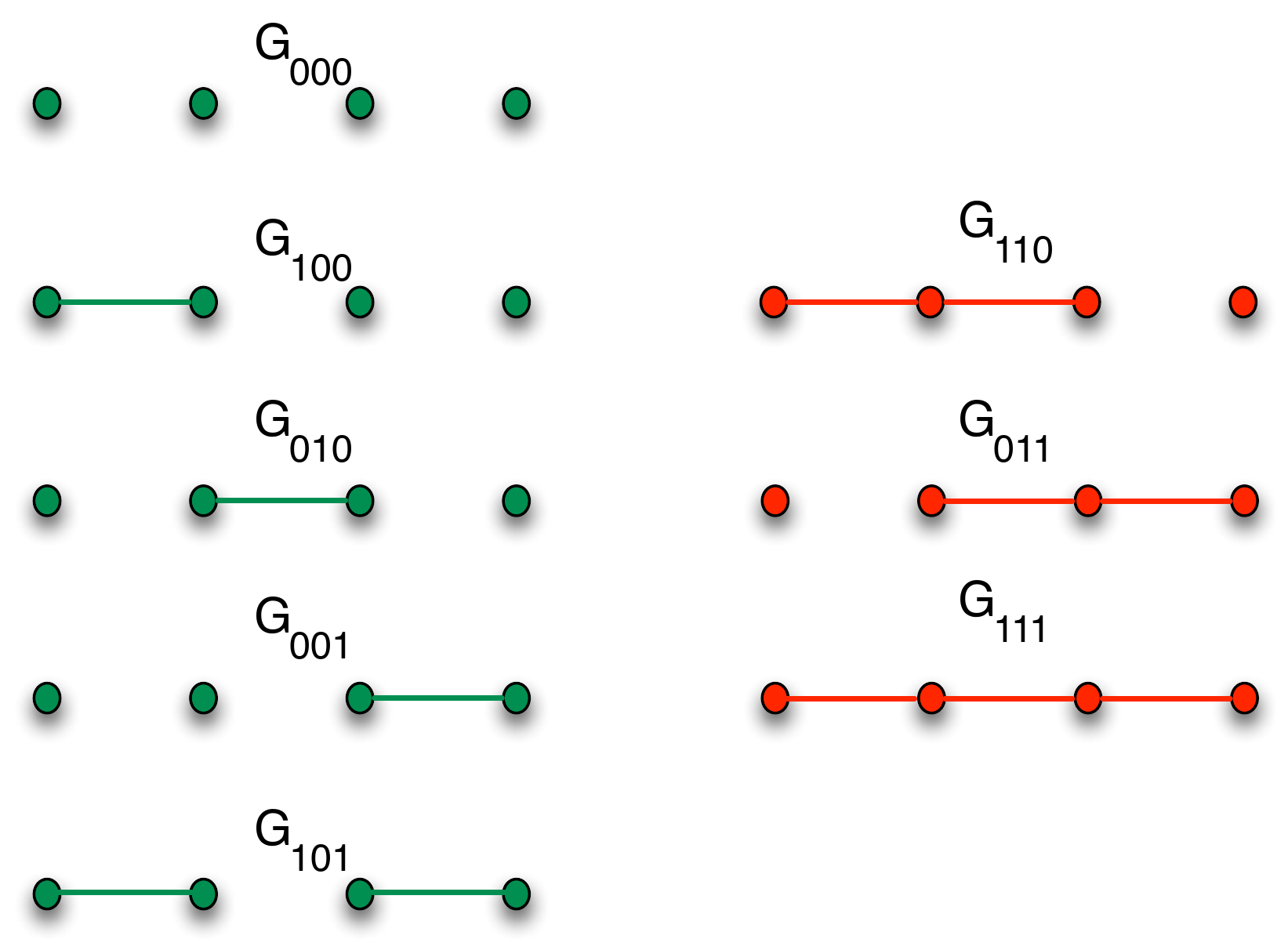}
\caption{No-error graphs (green, left)  and the error vectors (red, right) for the $(4,2)$ code.  $G_{000}$, $G_{100}$, $G_{010}$, $G_{001}$ and $G_{101}$ represent the no-error region  ${\mathcal R}_{w}^{c}$ and  $G_{110}$, $G_{011}$ and $G_{111}$ represent the error region, ${\mathcal R}_{w}$.}
\label{fig:42code-1}
\end{figure}

The error graphs of a $(4,2)$ code are represented in Fig.~\ref{fig:42code-1}. 
For $t \geq 3 t_{\trep}$, the volume of the error region is given by
$$\mbox{\upshape{vol} }\mathcal{R} = 24 t^2 t_{\text{rep}}^2-72 t_{\text{rep}}^3 t+64 t_{\text{rep}}^4.$$
\end{exe}
Details of the volume calculation are presented in Sec.~\ref{sec:ErrorRegions}.

\subsection{Set Avoidance Bounds for $(n,k)$ MDS codes}
\label{sec:avoidance:MDS}
For an $(n,k)$ MDS code, let $\alpha_{j}(n,k)$ denote the number of \emph{error} graphs labeled by error vectors $\bm{b}$ with Hamming weight $j$. We define the error polynomial as:

\begin{equation}
e(\rho)=\sum_{j=n-k}^{n-1}\alpha_{j}(n,k)v_{n-1-j,j}.
\label{eq:errorPolynomial}
\end{equation}
Let $\beta_{j}(n,k)$ be the number of no-error graphs of weight $j$ for an $(n,k)$ code,
$$\beta_j(n,k) = {n-1 \choose j} - \alpha_j(n,k),$$
where the term ${n-1 \choose j}$ is the total number of binary strings of length $n-1$ and Hamming weight $j$. Analyzing the labels $b_1 \ldots b_{n-1}$, it follows that $\beta_j(n,k)$ is the number of binary strings of length $n-1$ and weight $j$ that has no runs of $(n-k)$ or more $1s$. This number and its relation with generalizations of the Pascal Triangle was thoroughly studied in \cite{Bollinger1,Bollinger2}. It follows immediately that $\alpha_j(n,k) = 0$, for $j = 0, 1, \ldots, (n-k-1)$.

Combining two results from \cite[Thm. 3.3]{Bollinger1} and \cite[Eq. 3]{Bollinger2}, we have $\beta_j(n,n-k) = C_{n-k}(n-j,j)$, where $C_m(l,s)$ is the coefficient of $x^s$ in the expansion of the polynomial generating function $(1 + x + \ldots + x^{m-1})^l$. This leads to the following
\begin{lemma} The number of no-error graphs of Hamming weight $j$ of an $(n,k)$-code is given by
\begin{equation}\begin{split}
\beta_j(n,k) &= C_{n-k}(n-j,j) = \\ &= \sum_{i=0}^a (-1)^i {n-j \choose i} {{n-1-(n-k)i} \choose n-j-1},
\end{split}
\end{equation}
where $a = \min \{ n-j, \left\lfloor j/(n-k) \right\rfloor \}$.
\label{lemma:HammingWeight}
\end{lemma}
We are now in position to prove:
\begin{thm} The error polynomial for an $(n,k)$ MDS code satisfies
\begin{equation} e(\rho) = \frac{n!}{(k-1)!} \rho^k + O(\rho^{k-1}).
\end{equation}
\label{thm:errorPolynomial}
\end{thm}

\begin{proof}
When expressed as a polynomial in $\rho$, the volume polynomial is given by
$$e(\rho)=\sum_{j=0}^{{n}} b_s \rho^{s}.$$
From Remark~\ref{rmk:degreeV} which follows Thm.~\ref{thm:volPol},  $b_s = 0$ for $s = k+1, \ldots n$, i.e. each volume polynomial in Eq. \eqref{eq:errorPolynomial} has degree at most $k$. In fact, the only polynomial that  has degree $k$ is $v_{k-1,n-k}(\rho)$. Also from Remark~\ref{rmk:degreeV}, the coefficient of $\rho^k$ in $v_{k-1,n-k}(\rho)$ is $k! {n \choose k}$, whereas from Lemma \ref{lemma:HammingWeight}, $\alpha_{n-k}(n,k) = k$. Thus, the highest degree term of $e(\rho)$ is $\alpha_{n-k}(n,k)k! {n \choose k} \rho^k = n!/(k-1)! \rho^k$, from where the theorem follows. 
\end{proof}

\begin{cor} The volume of $\mathcal{R}$ for an $(n,k)$-code satisfies
\begin{equation} \vol \mathcal{R} = \frac{n!}{(k-1)!} t^k \trep^{n-k} + \sum_{s=0}^{k-1} a_s t^s \trep^{n-s},
\end{equation}
where $a_s$, $s = 0, \ldots, k-1$, are constants.
\label{cor:Cor4}
\end{cor}
\begin{rmk}  When $\trep/t$ is small ($\trep/t \to 0$), 
\begin{equation}
\vol \mathcal{R} \approx \frac{n!}{(k-1)!} t^k \trep^{n-k}.
\label{cor:volErroRegionAsympt}
\end{equation}
\end{rmk}

\subsection{Averaging the Set Avoidance Bound for Poisson Failures the Multiplicative Gap}

We now evaluate the bound for Poisson failures with rate $\lambda$ i.e. inter failure durations that are iid exponential with mean $1/\lambda$. We also evaluate the multiplicative gap between the set avoidance upper bound and the asymptotic result (\ref{eqn:asymptoticConst}).

\begin{thm}\label{prop:boundRj} Let $\mathcal{R}_j$ be the error region of a $(j,j-(n-k))$-code, $j \geq n-k+1$. The probability of data loss of an $(n,k)$ coded is bounded by
\begin{equation}
P (\mathcal{D}_t) \leq \sum_{j=n-k+1}^n {n \choose j} e^{-\lambda t(n-j)} (\lambda t)^j \left(\frac{\vol \mathcal{R}_j}{t^j}\right).
\label{eq:upperBoundRj}
\end{equation}
\end{thm}
\begin{proof}Let $w(t)$ denote the random variable associated to the weight, i.e. the number of disks that failed at least once within $[0,t]$. We have:
\begin{eqnarray*}
P(\mathcal{D}_t)  = \sum_{j=n-k+1}^{n} { n \choose j} (1-e^{-\lambda t})^{j} e^{-\lambda t (n-j)} P(\mathcal{D}_t | w(t) = j).
\label{eq:boundLotsOfCalculations}
\end{eqnarray*}
and the RHS of the above equation can be bounded by using lower-dimensional versions of Thm. \ref{thm:Bound}:

\begin{equation}
\begin{split} 
&P(\mathcal{D}_t^c | w(t) = j) \geq \left(1 - \frac{\vol R_j}{t^j} \right)^{\lambda^j t^j/(1-e^{-\lambda t})^j}
\end{split}.
\label{eq:boundRj}
\end{equation}

The proof now follows by bounding \eqref{eq:boundRj} using the fact that $(1-x)^r \geq 1 - rx$ for any real numbers $r,x$ such that $r \geq 1$ and $0 \leq x \leq 1$.
\end{proof}

\begin{cor} Let $P^{(u)}(\mathcal{D}_t)$ be the upper bound in \eqref{eq:boundRj}. The multiplicative gap between $P^{(u)}(\mathcal{D}_t)$ and $P(\mathcal{D}_t)$ satisfies
\begin{equation}
\lim_{t_\rep \to 0} \frac{P^{(u)}(\mathcal{D}_t)}{P(\mathcal{D}_t)} = (e^{-\lambda t} + \lambda t)^{k-1}
\label{eq:multiplicativeGap}
\end{equation}
In particular, when $k = 1$ the bound is asymptotically tight.
\end{cor}
\begin{proof} From Equation \eqref{eq:boundRj} and (\ref{cor:volErroRegionAsympt}), the ratio $P^{(u)}(\mathcal{D}_t)/t_\rep^{n-k}$ is well approximated by
\begin{equation*}
\begin{split} & \approx \sum_{j=n-k+1}^{n} { n \choose j} e^{-\lambda t (n-j)} \left(\frac{j! t^{j-(n-k)}}{(j-(n-k)-1)!t^j} \right) {(\lambda t)^j}, \\ &=
\frac{n!}{(k-1)!} \lambda^{n-k+1} t \\ & \times \left[ \sum_{j=n-k+1}^{n} { k-1 \choose n-j} e^{-\lambda t (n-j)} (\lambda t)^{j-(n-k+1)}\right]
\end{split}
\end{equation*}
and the approximation is tight as $t_\rep \to 0$. Using Theorem \ref{thm:asymptoticBehavior} and after some algebraic manipulation, we conclude \eqref{eq:multiplicativeGap}.
\end{proof}

\section{Volume Calculations for Ordered Sets with Constrained Differences}
\label{sec:ErrorRegions}
Both of the approaches presented for estimating the data loss probability, the direct approach of Sec.~\ref{sec:closedForm} and the bounds based on set avoidance presented in Sec.~\ref{sec:avoidance} ultimately rely on the methods for volume calculation presented in this section.
The calculations presented here are for an ordered $s$-tuple, where $s$ is a dummy variable, no longer necessarily associated with the number of disks failures in the interval $[0,t]$.  In order to apply the results to Sec.\ref{sec:closedForm},  $s$ will indeed represent the total number of failures that occur in the interval $[0,t]$, whereas in order to apply the results to Sec.~\ref{sec:avoidance}, $s$  will be replaced by $n$, the number of disks in the system. The results in this generalizes the formulas in \cite{CampVai1} for any $(n,k)$ and provides the exact behavior of such formulas.

The volume of the error region can be determined by splitting it into disjoint simplices. Since, by definition, the region $\mathcal{R}$ is symmetric with respect to different orderings of the failures, we have ${\vol \mathcal{R}  = s! \, \vol({R} \cap \mathcal{S}})$.  We can thus restrict our analyses to ordered vectors $x_1 \leq x_2 \ldots \leq x_s$. The volume of the regions restricted to the ordered simplex is now presented.

We first observe that $\vol \mathcal{R}_{\bm{b}}$ only depends on the weight (number of nonzero entries) of $\bm{b}$ (see Lemma~\ref{lemma:Auxiliar}  and the remarks that follow in the Appendix). Thus it suffices to study graphs of the form $G_{0^i1^j}$, where $j$ is the weight of the vector $\bm{b}$. We will work with  volume polynomials $v_{ij}(\rho)$, a scaled version of the volume of the region $\mathcal{R}_{0^i1^j}$, where for convenience we repeat that  $\rho = t/t_{\rep}$ and $v_{ij}(\rho)$ associated with $\mathcal{R}_{0^i 1^j}$ is given  by ${v_{ij}(\rho)=s! \vol G_{0^i 1^j}/t_{rep}^s}.$

We prove in Appendix \ref{app:proofs} that $\mathcal{R}_{0^i}$ is a simplex with volume $(t-it_\rep)^s/s!$, provided $t \geq i t_\rep$. Alternatively, $v_{i0} = (\rho-i)^s$. For instance $v_{00}(\rho) = \rho^s$ is the volume polynomial of the region with no constraints on the differences $x_{i+1} - x_{i}$. Since the union of a region such that  $x_{i+1} - x_{i} \leq t_\rep$ and another one such that $x_{i+1} - x_{i} \geq t_\rep$ gives a region with no constraints on $x_{i+1} - x_i$, we have the following ``difference'' identity:
$$v_{i,j}(\rho)=v_{i,j-1}(\rho)-v_{i+1,j}(\rho).$$

Summarizing, the following rules provide a systematic method for calculating the volume polynomials associated with any node in the supergraph.
\begin{itemize}
\item{(Shift)}  $v_{i+1,j}(\rho) = v_{i,j}(\rho-1)$, $j=0,1,2,\ldots$, 
\item{(First Difference)} $v_{i,j}(\rho)=v_{i,j-1}(\rho)-v_{i+1,j}(\rho)$,
\item{(Initial Condition)} $v_{00}=\rho^{s}$.
\end{itemize}

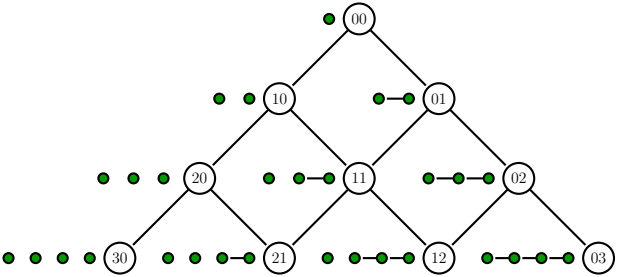
\begin{figure}[!htb]
\centering
\begin{tikzpicture}[-,>=stealth',shorten >=1pt,auto,node distance=3cm,
  thick,main node/.style={circle,fill=white!10,font=\sffamily\large\bfseries,scale=0.5}]

  \node[main node,draw] (1) {$00$};
  	  \node[main node,draw,scale=0.8,fill=black!40!green] (41) [left of = 1,xshift=2cm] {};
  \node[main node,draw] (2) [below left of=1] {${10}$};
  	  \node[main node,draw,scale=0.8,fill=black!40!green] (37) [left of = 2,xshift=2cm] {};
	  \node[main node,draw,scale=0.8,fill=black!40!green] (38) [left of = 37,xshift=2cm] {};
	  
  \node[main node,draw] (3) [below right of=1] {${01}$};
    	  \node[main node,draw,scale=0.8,fill=black!40!green] (39) [left of = 3,xshift=2cm] {};
	  \node[main node,draw,scale=0.8,fill=black!40!green] (40) [left of = 39,xshift=2cm] {};
	  
  \node[main node,draw] (4) [below left of=2] {${20}$};
	  \node[main node,draw,scale=0.8,fill=black!40!green] (28) [left of = 4,xshift=1.8cm] {};
	  \node[main node,draw,scale=0.8,fill=black!40!green] (29) [left of = 28,xshift=2cm] {};
	  \node[main node,draw,scale=0.8,fill=black!40!green] (30) [left of = 29,xshift=2cm] {};
  
  \node[main node,draw] (5) [below left of=3] {$11$};
  	  \node[main node,draw,scale=0.8,fill=black!40!green] (31) [left of = 5,xshift=2cm] {};
	  \node[main node,draw,scale=0.8,fill=black!40!green] (32) [left of = 31,xshift=2cm] {};
	  \node[main node,draw,scale=0.8,fill=black!40!green] (33) [left of = 32,xshift=2cm] {};
	  
  \node[main node,draw] (6) [below right of=3] {$02$};
	  \node[main node,draw,scale=0.8,fill=black!40!green] (34) [left of = 6,xshift=2cm] {};
	  \node[main node,draw,scale=0.8,fill=black!40!green] (35) [left of = 34,xshift=2cm] {};
	  \node[main node,draw,scale=0.8,fill=black!40!green] (36) [left of = 35,xshift=2cm] {};
	  
  \node[main node,draw] (7) [below right of=5] {$12$};
  \node[main node,draw] (8) [below left of=5] {$21$};
  \node[main node,draw] (9) [below left of=4] {$30$};
  \node[main node,draw] (11) [below right of=6] {$03$};
  \node[main node,draw,scale=0.8,fill=black!40!green] (12) [left of = 11,xshift=2cm] {};
  \node[main node,draw,scale=0.8,fill=black!40!green] (13) [left of = 12,xshift=2.1cm] {};
  \node[main node,draw,scale=0.8,fill=black!40!green] (14) [left of = 13,xshift=2.1cm] {};
  \node[main node,draw,scale=0.8,fill=black!40!green] (15) [left of = 14,xshift=2.1cm] {};
  
  \node[main node,draw,scale=0.8,fill=black!40!green] (16) [left of = 7,xshift=2cm] {};
  \node[main node,draw,scale=0.8,fill=black!40!green] (17) [left of = 16,xshift=2.1cm] {};
  \node[main node,draw,scale=0.8,fill=black!40!green] (18) [left of = 17,xshift=2.1cm] {};
  \node[main node,draw,scale=0.8,fill=black!40!green] (19) [left of = 18,xshift=2.1cm] {};
  
    \node[main node,draw,scale=0.8,fill=black!40!green] (20) [left of = 8,xshift=2cm] {};
  \node[main node,draw,scale=0.8,fill=black!40!green] (21) [left of = 20,xshift=2.1cm] {};
  \node[main node,draw,scale=0.8,fill=black!40!green] (22) [left of = 21,xshift=2.1cm] {};
  \node[main node,draw,scale=0.8,fill=black!40!green] (23) [left of = 22,xshift=2.1cm] {};

    \node[main node,draw,scale=0.8,fill=black!40!green] (24) [left of = 9,xshift=2cm] {};
  \node[main node,draw,scale=0.8,fill=black!40!green] (25) [left of = 24,xshift=2.1cm] {};
  \node[main node,draw,scale=0.8,fill=black!40!green] (26) [left of = 25,xshift=2.1cm] {};
  \node[main node,draw,scale=0.8,fill=black!40!green] (27) [left of = 26,xshift=2.1cm] {};

  \path[every node/.style={font=\sffamily\small}]
    (1) edge node [right] {} (2)
    (1) edge node [right] {} (3)
    (2) edge node [right] {} (4)
    (2) edge node [right] {} (5)
    (3) edge node [right] {} (5)
    (3) edge node [right] {} (6)
    (5) edge node [right] {} (7)
    (5) edge node [right] {} (8)
    (4) edge node [right] {} (8)
    (4) edge node [right] {} (9)
    (6) edge node [right] {} (7)
    (6) edge node [right] {} (11)
    (12) edge node [right] {} (13)
    (13) edge node [right] {} (14)
    (14) edge node [right] {} (15)
    (16) edge node [right] {} (17)
    (17) edge node [right] {} (18)
    (20) edge node [right] {} (21)
    (31) edge node [right] {} (32)
    (34) edge node [right] {} (35)
    (35) edge node [right] {} (36)
    (39) edge node [right] {} (40);

\end{tikzpicture}
\caption{Supergraph representation of the set of graphs $G_{0^i 1^{j}}$ for the $(4,2)$ code. A vertex with label $ij$ represents the graph $G_{0^i 1^{j}}$ (depicted on the left of each node of the supergraph). Note that the number of constraints
increases from zero at the top layer of the supergraph to three at the bottom layer of the supergraph.} 
\label{fig:graphp}
\end{figure}

The graphs $G_{0^i 1^j}$ are conveniently organized into a \emph{supergraph}, as illustrated in Fig. \ref{fig:graphp}, in order to facilitate computation of the volume polynomials. In this graph, each node is associated with a volume polynomial. For example, the top or root node in Fig.\ref{fig:graphp} is associated with the volume polynomial $v_{00}(\rho)$ and the polynomial associated with the graph  $G_{011}$ is $v_{12}$. 

We revisit the $(4,2)$ MDS code and compute the volume of the error region.
\begin{exe}[$(4,2)$-Code]
The  error vectors of a $(4,2)$ code are represented in Fig.~\ref{fig:42code-1}. Summing the volume polynomial corresponding to all error vectors and considering all orderings of the vector $(x_1,\ldots,x_n)$ we obtain
\begin{eqnarray}
\frac{1}{t_{rep}^4}\vol {\mathcal R}_w  & = &  2 v_{12}(\rho)+  v_{03}(\rho) \stackrel{(a)}{=}  v_{00}-v_{10}-v_{20}+v_{30} \nonumber \\
& = & \rho^4 - (\rho-1)^4 - (\rho-2)^4 + (\rho-3)^4  \nonumber \\ &=& 24 \rho^2 - 72 \rho + 64,
\end{eqnarray}
where in $(a)$ we applied the first difference rule and $(b)$ is a combination of the shift rule and the initial condition. This gives us, for $t \geq 3 t_{\trep}$,
$$\mbox{\upshape{vol} }\mathcal{R} = 24 t^2 t_{\text{rep}}^2-72 t_{\text{rep}}^3 t+64 t_{\text{rep}}^4.$$
\label{ex:one}
\end{exe}

In the following two examples, we calculate $\sum_{\bm{b} \in {\mathcal B}_{\bm{f}}} v_{\bm{b}}(\rho)$ 
in (\ref{eqn:thmMain}), related to the direct calculation of the data loss probability.
\begin{exe}
Suppose $\bm{m}=(1,1,1,1)$ and a $(4,2)$ MDS code is used. Consider the failure pattern $\bm{f}=1234$. Then  $B_{\bm{f}}=\{ 2(0^1 1^2),0^01^3 \}$ and $s=4$.
From this we can write down the volume polynomial as $2v_{12}(\rho)+v_{03}(\rho)$.
Upon simplification we obtain $v_{00}(\rho) -v_{10}(\rho)-v_{20}(\rho) +  v_{30}(\rho)=\rho^4-(\rho-1)^4-(\rho-2)^4+(\rho-3)^4= 24\rho^2-72\rho+64$. 
\label{ex:two}
\end{exe}
\begin{rmk}
Observe that the volume polynomials for Ex.~\ref{ex:one} and Ex.~\ref{ex:two} are identical.
\end{rmk}
Another example related to the direct calculation.
\begin{exe}
Suppose $\bm{m}=(2,2,1,1)$ and a $(4,2)$ MDS code is used. Consider the failure pattern $\bm{f}=121234$. Then  $B_{\bm{f}}=\{ 2(0^1 1^2),0^01^3 \}$ and $s=6$.
From this we can write down the volume polynomial as $2v_{12}(\rho)+v_{03}(\rho)$.
Upon simplification we obtain $v_{00}(\rho) -v_{10}(\rho)-v_{20}(\rho) +  v_{30}(\rho)=\rho^6-(\rho-1)^6-(\rho-2)^6+(\rho-3)^6= 60\rho^4-360\rho^3+960\rho^2-1260\rho+664$. 
\end{exe}

The following lemma uses the aforementioned rules to provide closed form expressions for $v_{ij}(\rho)$. 
\begin{thm}
The volume polynomial $v_{ij}(\rho)=\sum_{r}a_{r} \rho^{r}$ satisfies the following properties
\newline
(i)
\begin{equation}
v_{ij}(\rho)=\sum_{l=0}^j (-1)^{j-l} {j \choose l} (\rho-i-j+l)^s.
\label{eq:ShiftForm}
\end{equation}
(ii) \begin{equation}
a_{r}= {s \choose r} j! (-1)^{s-r+j} \left( \sum_{m=0}^{s-j}{s-r \choose m} i^{m} S(s-r-m,j) \right),
\end{equation}
where $S(l,m)$ is a Stirling number of the second kind (see, e.g., \cite{Cameron:1994}). 

\label{thm:volPol}
\end{thm}
\begin{proof}
Given a function $f(x)~:~\mathbb{R} \rightarrow \mathbb{R}$, define the shift operator $S(f(x)):=f(x-1)$ and the first difference operator $\Delta(f(x)):=f(x)-f(x-1)$. Then (i) follows from the observation that $v_{ij}(\rho)=S^i\Delta^j (\rho^s)$.  Write $S=(1-\Delta)$ in order to
express the operator in terms of powers of $\Delta$. This gives 
\begin{equation}
v_{ij}(\rho)=\left(\sum_{l=0}^i (-1)^{i-l}{i \choose l} \Delta^{i+j-l} \right)(\rho^s).
\end{equation}
To prove (ii), expand the last term in (\ref{eq:ShiftForm}) and interchange the order of summation, so that
\begin{equation}
v_{ij}(\rho)=\sum_{m=0}^{n} \rho^{m} {s \choose m} (-1)^{s-m} \sum_{l=0}^{j} (-1)^{l} {j \choose l}(i+l)^{s-m}.
\end{equation}
The result follows directly by further expanding the last term in the above equation and from an identity for Stirling numbers of the second kind (e.g. Prop. 5.3.5, \cite{Cameron:1994}).
 \end{proof}

\begin{rmk}
 ${deg(v_{ij}(\rho))=s-j}$ and $a_{s-j}=j! {s \choose j}$.
\label{rmk:degreeV}
\end{rmk}

%
%


\section{Numerical and Simulation Results}
\label{sec:Numerical}
\begin{figure}[!htb]
\centering
\includegraphics[scale=0.53]{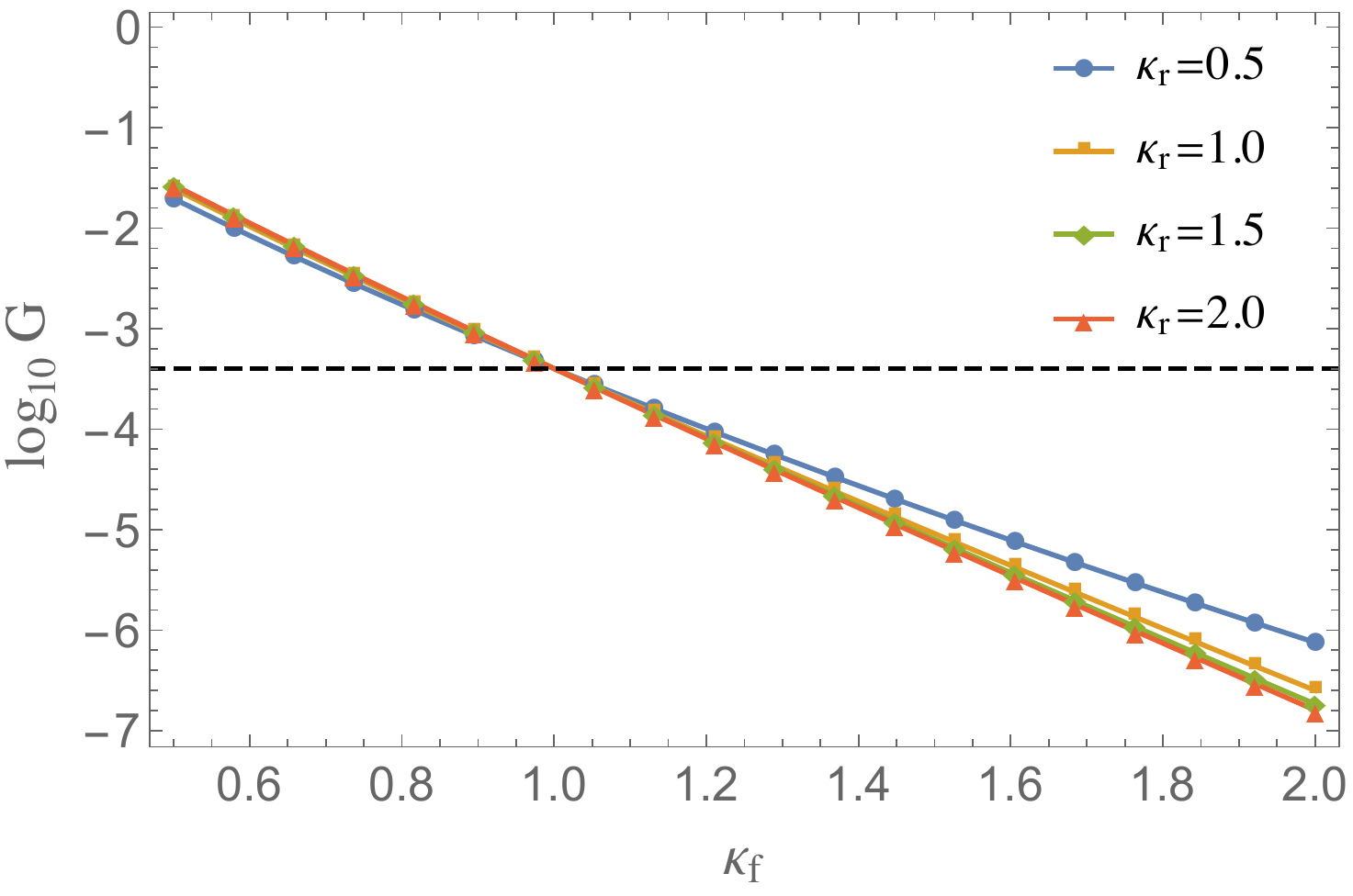}
\hfill
\includegraphics[scale=0.53]{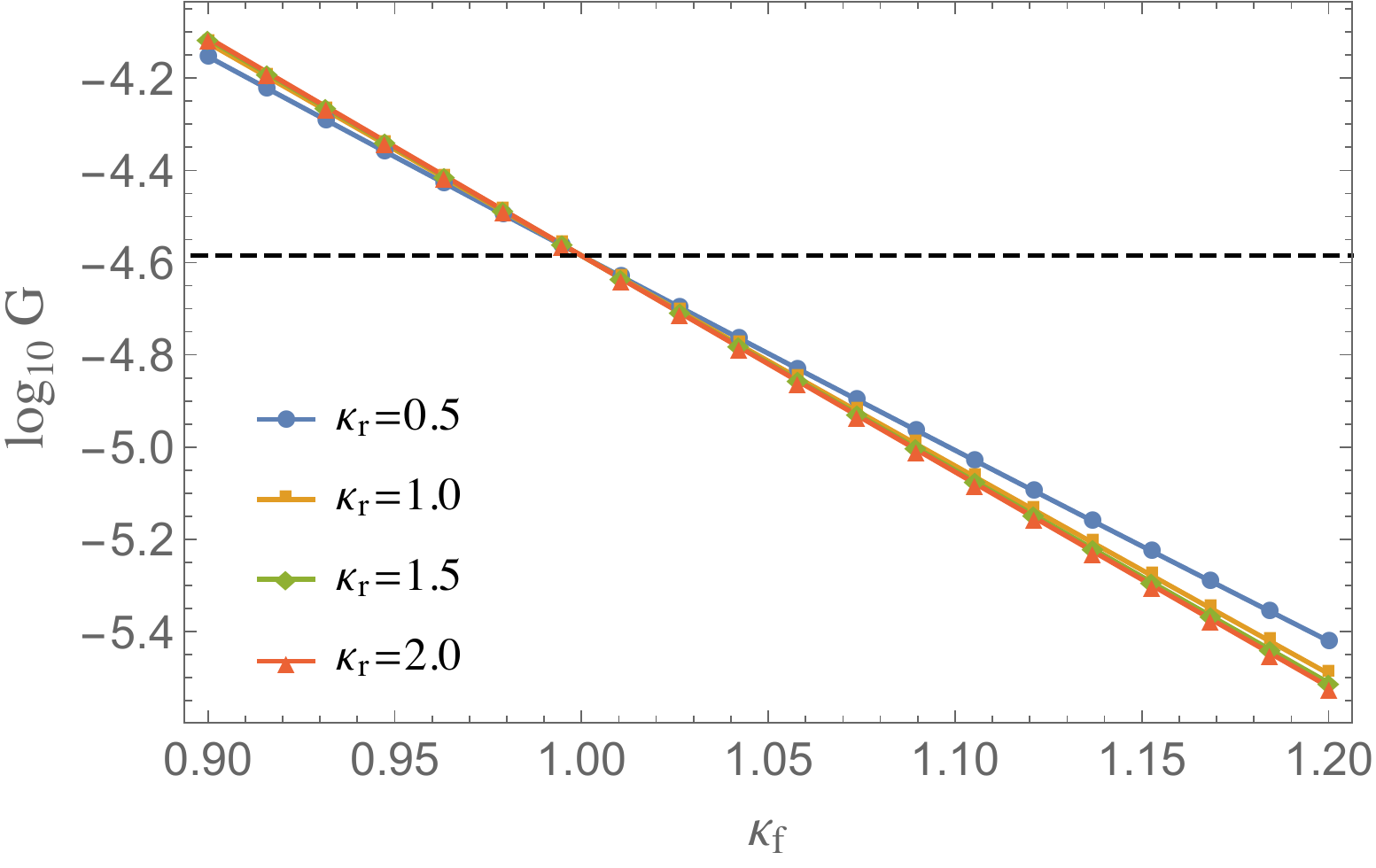}
\caption{$G=P(Y<Z)$ for Weibull distributed failures and repairs for  $(1/\lambda,1/\mu)=(0.25,10^{-4})$ (left), and $(1/\lambda,1/\mu)=(461386,12)$ (right).  The horizontal line is the value of $G$ for exponentially distributed repair and failure distributions. }
\label{fig:Weibull}
\end{figure}

\begin{table*}[!htb]
\centering
\begin{tabular}{cccccccccc}
$n$ & $k$ & $t$ & $\kappa_f$ & $\kappa_r$ & $1/\lambda$ & $1/\mu$ & $P(D_t)$ (\ref{eq:LimitingForm}) & $P(D_t)$ (Simulation) & Standard Deviation \\ \hline
4 & 2 & 1.0 & 1.5 & 2.0 & 0.1 & 0.001 & $3.343\times 10^{-6}$ & $3.429 \times 10^{-6}$ & $4.07 \times 10^{-7}$ \\ 
4 & 2 & 1.0 & 0.75 & 2.0 & 0.1 & 0.001 & 0.0044 & 0.0044 & $5.38 \times 10^{-4}$ \\
4 & 2 & 1.0 & 0.75 & 0.75 & 0.1 & 0.001 & 0.0035 & 0.0036 & $1.94 \times 10^{-4}$ \\
4 & 2 & 1.0 & 0.75 & 0.75 & 0.1 & $10^{-6}$ & $1.185 \times 10^{-7}$ & $1.221 \times 10^{-7}$ & $1.22 \times 10^{-8}$ \\ 
8 & 5 & 1.0 & 0.75 & 1.25 & 0.001 & $10^{-6}$ & $8.9289 \times 10^{-5}$ & $8.8383 \times 10^{-5}$ & $1.2397 \times 10^{-5}$ \\ 
8 & 5 & 1.0 & 2.0 & 2.0 & 0.01 & $0.001$ & $3.981 \times 10^{-5}$ & $4.012 \times 10^{-5}$ & $1.548 \times 10^{-6}$ \\ 
8 & 5 & 1.0 & 0.5 & 2.0 & 0.01 & $10^{-6}$ & $1.013 \times 10^{-4}$ & $1.008 \times 10^{-4}$ & $2.766 \times 10^{-6}$ \\
\hline
\end{tabular}
\caption{Validation of Eqn. (\ref{eq:LimitingForm}) through simulation. $t$ is the observation time window, $(\kappa_f,1/\lambda)$ and $(\kappa_r,1/\mu)$ are the (shape,mean) values for the Weibull distributed failure and repair durations, resp.  }
\label{table:SimulationResults}
\end{table*}
The  data loss probability expression (\ref{eq:LimitingForm}) for general distributions was validated through simulation for the family of Weibull distributions, 
with pdf given by
\begin{equation}
f(x) = \frac{\kappa}{\alpha}\left(\frac{x}{\alpha}\right)^{\kappa-1}e^{-(x/\alpha)^{\kappa}},
\end{equation}
where $\kappa$ is the \textit{shape parameter} and $\alpha$ is the \textit{scale parameter}. Its mean is given by $\alpha \Gamma(1+1/\kappa)$, where $\Gamma$ is the Gamma function. The exponential distribution is a special case of this family, for $\kappa = 1$ and $\alpha=1/\lambda$.
A few results are tabulated in Table~\ref{table:SimulationResults}. The agreement between theory and simulation is seen to be close. The simulator used here generates a sequence of failure and repair durations according to the specified distribution, calculates runs of the event $Y<Z$ and for each run generates disk labels drawn uniformly and iid on $\{1,2,\ldots,n\}$. The standard deviation of the estimates reported in Table~\ref{table:SimulationResults} was varied by changing the number of independent experiments.

The probability  $G=P(Y<Z)$ (defined immediately following (\ref{eqn:JointProb})) that a failure will occur prior to the completion of a repair is sensitive to the shape of the distribution.
This is illustrated, in Fig.~\ref{fig:Weibull}, for the family of Weibull distributions, 
In the plots below, it is assumed that failures are Weibull distributed with mean failure duration set to $1/\lambda=461386$ and mean repair duration $t_\rep = 1/\mu$ set to $12.0$, values that were obtained by Elerath and Pecht in an experimental study of disk failures~\cite{elerath2007enhanced}. For a fixed value of $\kappa_r$, the shape parameter for the Weibull repair duration distribution with mean $t_\rep$, the shape parameter $\kappa_f$ of the failure duration is varied, and the scale parameter is adjusted to keep the mean constant at $1/\lambda$.  As $\kappa_f$ varies the value of $G$ is seen to change significantly. This plot also shows the value of $G$ for exponentially distributed failure and repair durations. The value of $G$ is seen to be less sensitive to the shape parameter of the repair distribution.

\begin{figure}[!htb]
\centering
\includegraphics[scale=0.50]{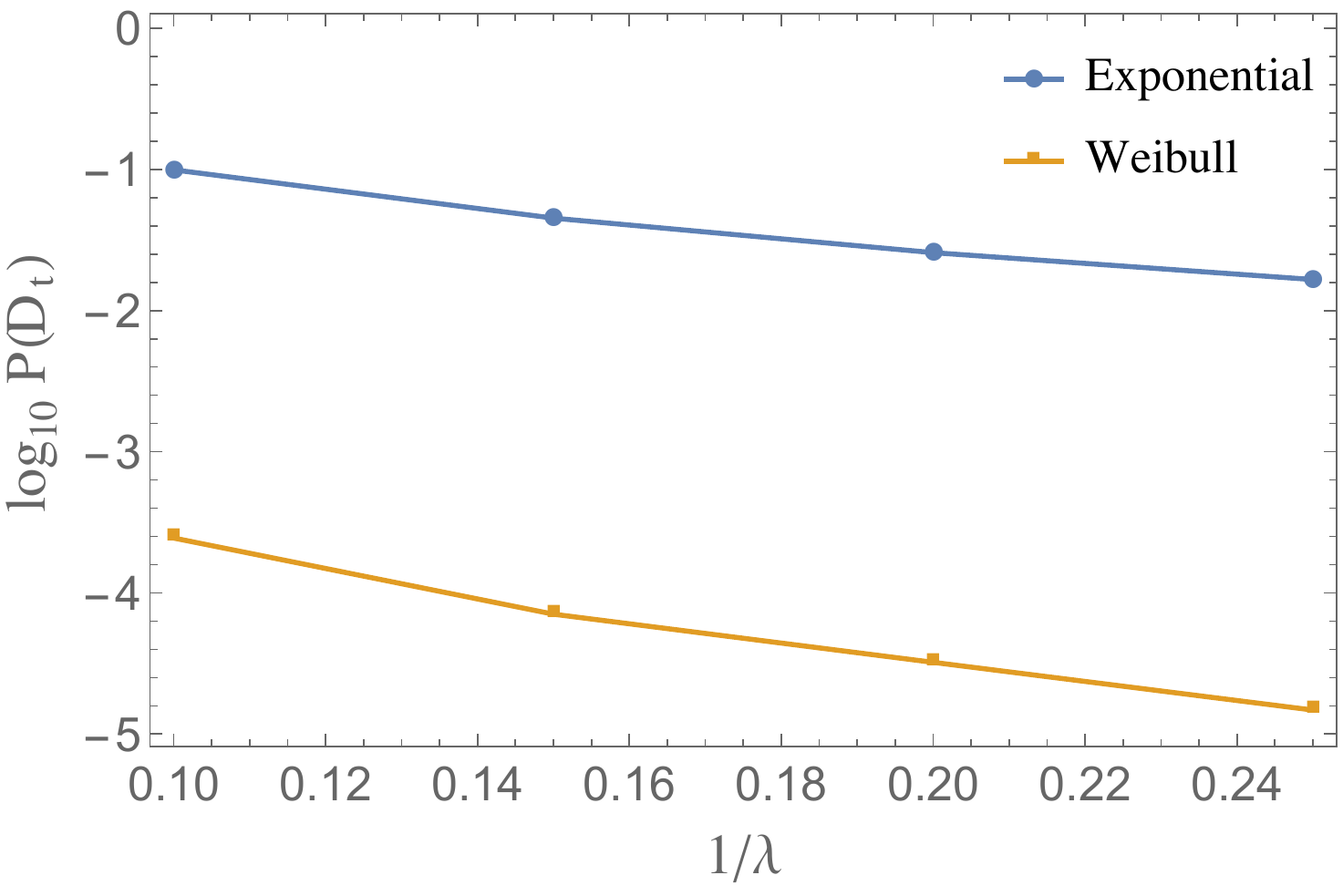}
\caption{Probability of data loss for time window of duration $t=1$,  when (failures, repairs) are both Weibull with $(\kappa_f,\kappa_r)=(0.5,2.0)$ and when (failures, repairs) are exponentially distributed. In all cases the mean repair duration is  $10^{-4}$. The scale parameters of the Weibull distributions are adjusted to keep the mean failure duration $1/\lambda$ and mean repair duration constant. }
\label{fig:CodeWeibullExp}
\end{figure}

In Fig.~\ref{fig:CodeWeibullExp} the performance of a $(4,2)$ code is compared for Weibull and exponential distributions for fixed mean repair and failure durations. As already observed, the value of $G$ is seen to depend on the shape parameter, and the impact on the data loss probability is magnified by redundancy $(n-k)$ of the code. The predicted gap in reliability is verified by the simulation.
The impact on the reliability of the system especially for a powerful code can be quite dramatic---an order of magnitude difference in $G$ becomes ten orders of magnitude for a code with $10$ check symbols.  It is also clear that the mean time between failures for individual disks is not a sufficient determinant of overall system reliability.

\begin{figure}[htb]
\centering
\includegraphics[width=3.0in]{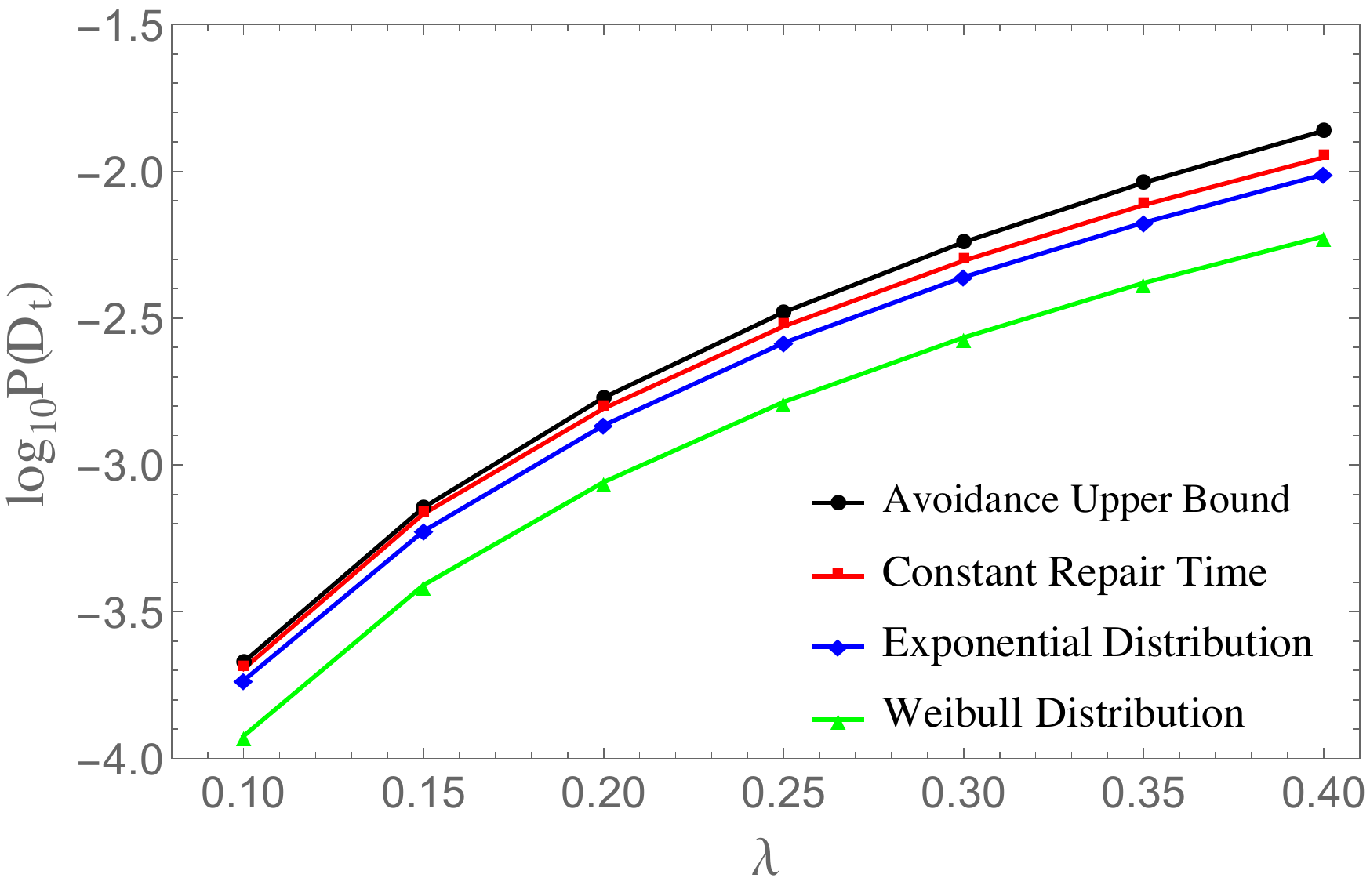}
\caption{Simulation results for exponential, constant and Weibull distributions, for $(n,k)=(4,2)$. Here $t = 1$, the mean repair time is set to $0.1$, and the mean time to failure is assumed to be exponential with parameter $\lambda$ between $0.1$ and $0.4$. For the Weibull distribution, $\kappa = 0.5$ and $\beta = t_\rep/2$ (so its mean is $t_\rep$). }
\label{fig:ComparisonsToConstant}
\end{figure}

In Fig.~\ref{fig:ComparisonsToConstant} we show that constant repair duration represents the worst case among the distributions considered. The simulation was carried out for a $(4,2)$-code, for fixed mean repair time ($t_{\rep}=0.1$), and exponential failure times with mean $1/\lambda$. For the Weibull distribution, the shape parameter was chosen to be equal to $0.5$ and the scale parameter $t_\rep/2$, so that the mean equals $t_{\rep}$. Simulations were based on $10^{7}$ samples for each value of $\lambda$, using the algorithm of \cite{sasidharan2013high} for the failure process.

Also in Fig~\ref{fig:ComparisonsToConstant}, for the {constant-repair-duration} case, we compared the simulation results and the upper bound, and a good agreement between both in cases where $\lambda t$ is close to unity. It is to be noted that this is the case in many practical situations. We stress the fact that in this case we have a \textit{true} upper bound on the reliability, whereas other methods proposed in the literature only provide estimates. We also caution the reader that the gap between the upper bound and the simulation results can be large when the product $\lambda t$ is large, as we have observed earlier in our discussion about the multiplicative gap.

\begin{figure}[htb]
\centering
\includegraphics[width=3.0in]{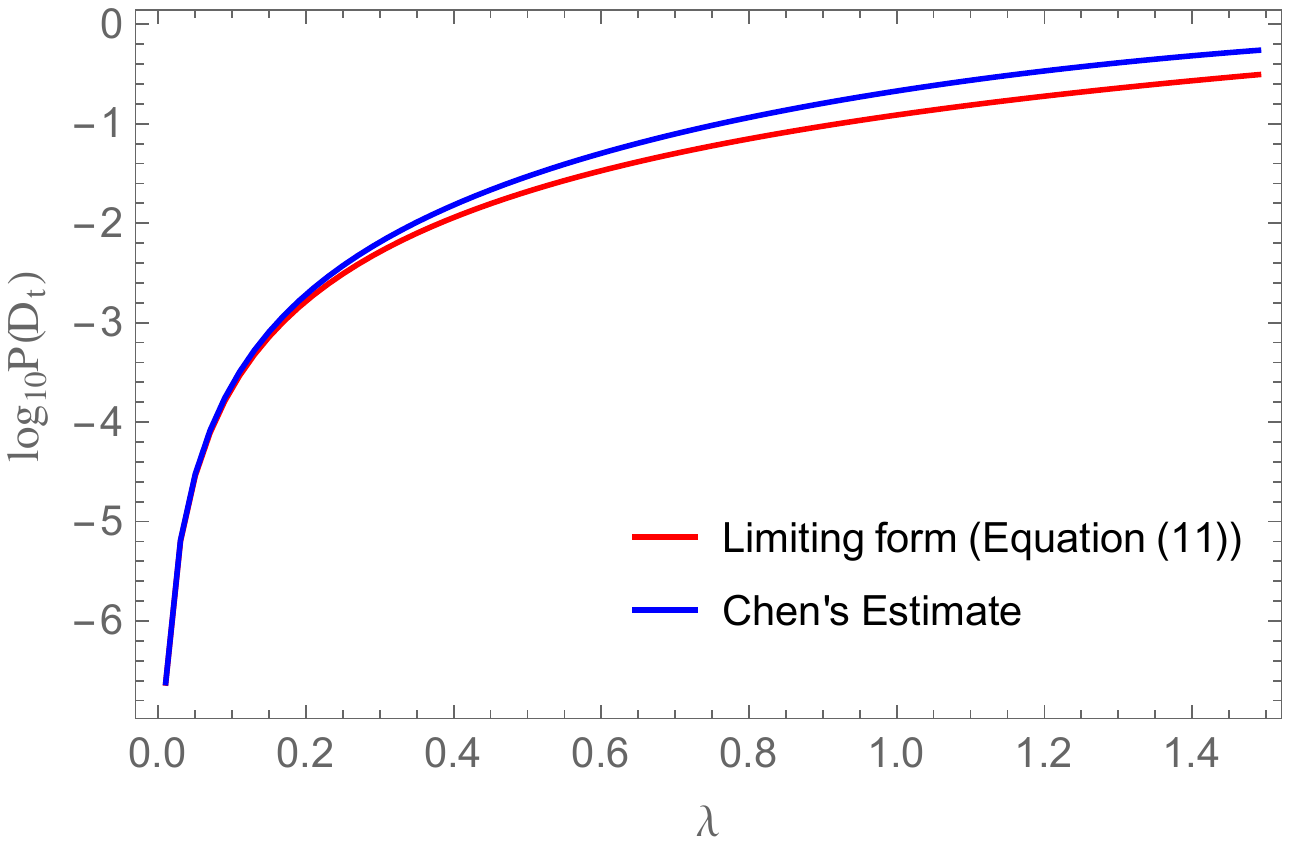}
\caption{Comparison between the limiting form (\ref{eq:LimitingForm}) for the exponential distribution and Chen's estimate with mean repair time $1/\mu = t_\rep = 0.1$, $t = 1$ and $(n,k)=(4,2)$.}
\label{fig:boundVsChen}
\end{figure}

For exponential failure and repair duration distributions it is shown in ~\cite{resch2011development}, based on the results in \cite{Chen:1994} (see also \cite[Eq. 6.69]{VenkThesis2012}), that the average time until a data loss event for an $(n,k)$ MDS erasure code is given by
\begin{equation} MTTDL_c = \frac{1}{\lambda^{n-k+1} t_\rep^{n-k}} \frac{(k-1)!}{n!}.
\end{equation}
From this we get the approximation $\hat{P}(D_t) = 1- \hat{R}(t) = 1 - e^{-t/MTTDL_c}$.  A first order approximation of $P(D_t)$ is
\begin{equation}
\hat{P}(D_t) \approx \frac{n!}{(k-1)!} \lambda^{n-k+1} t_\rep^{n-k} t
\label{eq:gapchen}
\end{equation}
A comparison of (\ref{eq:LimitingForm}) for exponential failure and repair distributions and (\ref{eq:gapchen}) is provided in Fig~\ref{fig:boundVsChen}. While there is close agreement in the limit, a deviation is observed for moderately large probabilities.

\section{Summary, Conclusions and Future Work}
\label{sec:Conclusions}

We have addressed the problem of directly evaluating the probability of data loss in an erasure coded distributed data storage system. A formula is derived for general  iid failure and repair duration distributions using combinatorial methods. For the case where  the repair duration is constant, we develop a  combinatorial-geometric approach that enables us to directly calculate and bound the data loss probability, in contrast to widely used methods that estimate the integral of the reliability function.  Further, our analysis is more refined, in the sense that we are able to derive expressions for the data loss probability conditioned on the number of failures in a given time window.

Our analytic results for general distributions indicate that $G$, the probability that a failure duration is smaller than a repair duration, is sufficient  for characterizing the data loss probability for highly reliable systems. In particular, distributions with the same mean failure and repair durations are seen to exhibit a wide range of values for $G$ and hence data loss probability.
This provides motivation for studying, in addition to erasure coding strategies, the impact of networking technologies, and network design, as well as the impact of physical component design (mechanical components in disk drives) with $G$ as a figure of merit. 

Finally, we mention that the set avoidance bound misses some of the subtle correlations between $n$-tuples of failure times. Analytic methods for taking these correlations into account should help to improve this bound. We leave this to future work.

\section{Acknowledgement}
We thank the reviewers and the AE for carefully reading the manuscript and for the numerous suggestions that have resulted in many improvements. This work was begun while the second author was visiting AT\&T Labs-Research in 2013.
The authors acknowledge with appreciation Prof. Sueli Costa for enabling the visit of the first author (VV) to Campinas in 2014.

\begin{appendices}
\section{Combinatorial Preliminaries}
\label{app:proofs}

\subsection{Properties of the Error Polytope (Sec. \ref{sec:ErrorRegions}) }
We now provide formal justification of the general rules for calculating the volume polynomial $v_{ij}(\rho)$.
We start with some observations on the error region and error graphs.

\begin{lemma}
Let $j_0, \ldots, j_i$ be integers such that $0 = j_0 < j_1 < j_2 < \ldots < j_i < s$. Consider the region
\begin{equation} \mathcal{R} = \left\{ (x_1,\ldots,x_s) : \begin{array}{c} 0 \leq x_1 \leq \ldots \leq x_s \leq t \\
x_{j_l+1} - x_{j_l} \geq t_{\text{rep}}, \,\, l = 1, \ldots, i \\
 \end{array} \right\}.
\end{equation} 
We have $\vol \mathcal{R} = (t-i t_\rep)^s/s!$. 
\label{lemma:Auxiliar}
\end{lemma}
\begin{proof}
Consider the translation $\phi(\bm{x})= \bm{x} - \bm{u}$, where $\bm{u}=(u_1,u_2,\ldots,u_s)$ defined as
$${u}_i = l t_\rep\mbox{, if }i = j_{l}+1,\ldots,j_{l+1}.$$
Let $y = \phi(\bm{x})$. The translated region $\phi(\mathcal{R})$ is given by:
\begin{equation*}
\phi(\mathcal{R}) = \left\{(y_1,\ldots,y_s) : \begin{array}{c} 0 \leq y_1 + u_1 \leq \ldots \leq y_s + u_s \leq t \\
y_{j_l+1} - y_{j_l} \geq 0, \,\, l = 1, \ldots, i \\
 \end{array} \right\}.
\end{equation*}
Eliminating redundant inequalities we obtain
\begin{equation*}
\phi(\mathcal{R}) = \left\{ (y_1,\ldots,y_s): 0 \leq y_1 \leq y_2 \leq \ldots \leq y_s \leq t-it_\rep\right\}.
\end{equation*} 
This last set of inequalities corresponds to a well-known regular simplex whose volume is $(t-it_\rep)^s/s!$, concluding the proof.
\end{proof}
In particular, Lemma~\ref{lemma:Auxiliar} shows that the volume of a polytope ${\mathcal R}_{\bm{b}}$ defined by an vector $\bm{b}$, depends only on its weight.
\begin{lemma}
Let $1 \leq i \leq s-1$. The volume polynomial associated with the $i$th node along the left boundary of the super-graph is given by:
\begin{equation}v_{i0}(\rho)= (\rho- i)^s. \label{eq:leftBoundary}
\end{equation}
\label{lemma:leftedge}
\end{lemma}
\begin{proof} Recall that, by definition, $v_{i0}(\rho)= s! \vol G_{0^i}/t_\rep^s$, where $\rho = t/t_\rep$. Thus the statement is equivalent to $\vol G_{0^i} = (t-i t_\rep)^s/s!$, which, in turn, is a special case of Lemma \ref{lemma:Auxiliar}, for $j_l = l, l = 1, \ldots, i$.
\end{proof}

\section{Set Avoidance Bounds}
\label{app:setAvoidance}
\textit{Proof of Lemma \ref{lem:avoidance}:}
\begin{eqnarray*}
\lefteqn{P\left({\mathcal X}\times {\mathcal Y} \bigcap \mathcal{R} =\emptyset\right)} \nonumber \\
&  = &  E\left(P\left({\mathcal X} \bigcap \bigcup_{i=1}^{m_2} \mathcal{R}(y_{i}))=\emptyset~|~{\mathcal Y}\right)\right) \nonumber \\
& =  & E\left(P\left(X  \notin \bigcup_{i=1}^{m_2} \mathcal{R}(y_{i}))~|~{\mathcal Y}\right)^{m_1}\right) \nonumber \\ 
&\stackrel{(a)}{\geq} &  E\left(P\left(X  \notin \bigcup_{i=1}^{m_2} \mathcal{R}(y_{i}))~|~{\mathcal Y}\right)\right)^{m_1} \nonumber \\
& = &  \left( P(\{X\} \times {\mathcal Y} \bigcap \mathcal{R}=\emptyset)\right)^{m_1},
\end{eqnarray*}
where in (a) we have used Jensen's inequality. The condition for equality follows directly from the condition for equality in Jensen's inequality. \\ \qed

In general we do not expect the condition for equality to hold, except in the case where one of the random sets has a single element. 

For the next upper bound, we use the following generalized version of the union bound: if $A_1,A_2, \ldots A_m$ are $m$ events, then the probability of $\cup_{i=1}^m A_i$ is lower bounded by
\begin{equation} P\left(\bigcup_{i=1}^m A_i\right) \geq \sum_{i=1}^m P(A_i) - \sum_{j=i+1}^m\sum_{i=1}^m P(A_i \cap A_j).
\label{eq:generalUnion}
\end{equation}
In what follows we denote the event $\left\{ (X,Y)\in \mathcal{R}\right\}$ by  $\varepsilon(X,Y)$.
\begin{thm}  Let $Q_1(x) = P(\varepsilon(x,Y))$ and $Q_2(y) = P(\varepsilon(X,y))$. The set avoidance probability is upper bounded by
\begin{equation*}
\begin{split} &P\left({\mathcal X}\times {\mathcal Y} \bigcap \mathcal{R} =\emptyset\right)  \leq 1 - m_1 m_2 P(\varepsilon(X,Y)) + \\ &+ 2 {m_1 \choose 2}{m_2 \choose 2} P(\varepsilon(X,Y))^2 + \\ &+ m_2 {m_1 \choose 2} E[Q_1(X)^2] + m_1 {m_2 \choose 2} E[Q_2(Y)^2]
\end{split}
\end{equation*}
\label{thm:upperbound}
\end{thm}
\begin{proof} First note that
\begin{eqnarray}
P\left({\mathcal X}\times {\mathcal Y} \bigcap \mathcal{R} = \emptyset \right) &=& 1 - P\left(  \bigcup_{j=1}^{m_2} \bigcup_{i=1}^{m_1} \varepsilon(X_i,Y_j) \right). \nonumber \\
\label{eq:rewriteAvoidance} \end{eqnarray}
From Eq. \eqref{eq:generalUnion}, the RHS of \eqref{eq:rewriteAvoidance} can be lower bounded 
\begin{eqnarray*}
1 - m_1 m_2 P\left(  \varepsilon(X,Y) \right) + \sum P \left( \varepsilon(X_i,Y_j) \cap \varepsilon(X_{i'},Y_{j'}) \right),
\end{eqnarray*}
where the summation is over all ${m_1 m_2 \choose 2}$ distinct choices of cross terms $\varepsilon(X_i,Y_j) \cap \varepsilon(X_{i'},Y_{j'})$. 
Now, for the probability of the cross terms, we have three cases. If $i \neq i'$ and $j \neq j'$ then, due to independence, $P \left( \varepsilon(X_i,Y_j) \cap \varepsilon(X_{i'},Y_{j'}) \right) = P\left( \varepsilon(X,Y) \right)^2.$
On the other hand, if $i = i'$ (and $j \neq j'$) let $f(x_i) = P( \varepsilon(X_i,Y_j) \cap \varepsilon (X_i,Y_{j}) | X_i = x_i) = Q_1(x_i)^2$. Then:
$$P \left( \varepsilon(X_i,Y_j) \cap \varepsilon(X_{i},Y_{j'}) \right) = E[ f(X) ] = E[ Q_1(X)^2].$$
The case $j = j'$ is analogous. Counting the number of occurrences of the three cases leads us to the theorem.
\end{proof}
 
If $X$ and $Y$ are uniformly distributed over a set $\mathcal{S} = \mathcal{S}_1 = \mathcal{S}_2$, the functions $Q_1, Q_2$ have a natural geometric interpretation, as can be seen in the next example.

\begin{exe} \label{ex:1} Let $\mathcal{R} = \left\{(x,y) \in [0,1]^2 : |x-y| \leq t_{\rep} \right\}$ be the error region of a $(2,1)$-code, and consider $X$ and $Y$ uniformly distributed over $[0,t]$. Then $P(\varepsilon(X,Y)) = (2 t_{\rep}t - t_{\rep}^2)/{t^2}.$
The function $Q_1(x)$ corresponds to the probability that $Y$ belongs to the shadow of $\mathcal{R}_1(x)$ on the $y$-axis, which, in this case, is the length of $\mathcal{R}_1(x)$. One can easily see that $Q_1(x) \leq 2 t_{\rep}$, thus $E[Q_1(X)^2] \leq 4t_{\rep}^2$. By symmetry, $E[Q_2(X)^2] \leq 4t_{\rep}^2$.

\begin{figure}[!htb]
\centering
\includegraphics[scale=0.45]{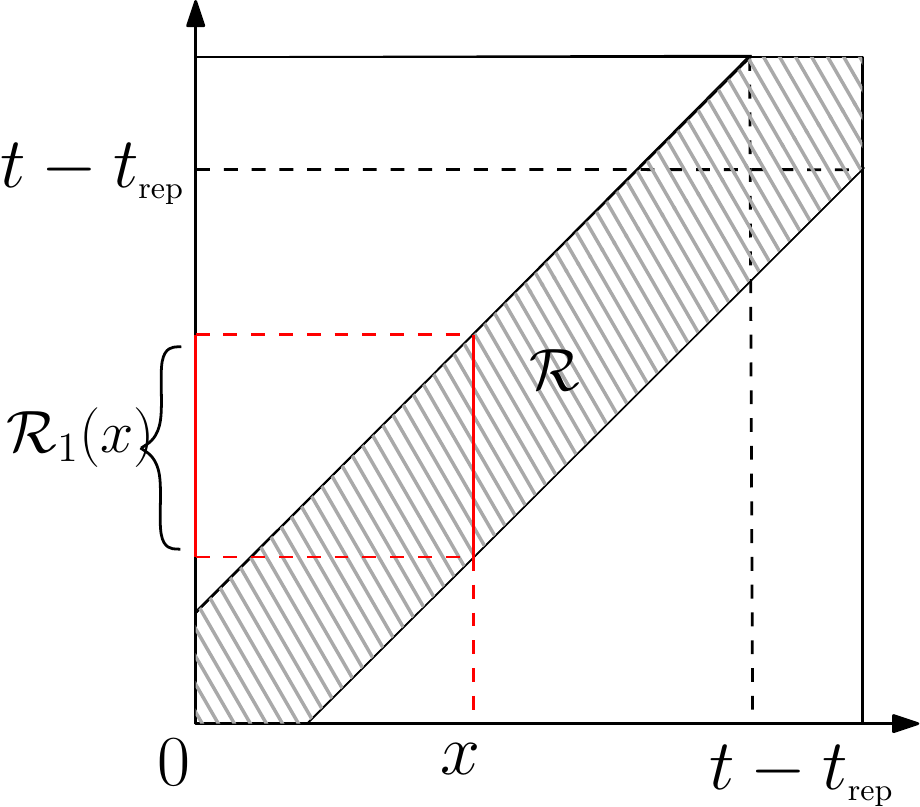}
\caption{Region $\mathcal{R}$ and the shadow of $\mathcal{R}_1(x)$ on the $y$-axis}
\label{fig:ex1}
\end{figure}

Assume that $X_1, \ldots, X_m$ and $Y_1, \ldots, Y_m$ are iid with the same distribution as $X$ and $Y$. Applying Corollary \ref{Cor:main}, we have
$$P\left({\mathcal X}\times {\mathcal Y} \bigcap \mathcal{R} =\emptyset\right) \geq (1-t_{\rep}/t)^{2m_1m_2} \geq 1 - 2m_1 m_2 \frac{t_{\rep}}{t}.$$
On the other hand, Thm. \ref{thm:upperbound} together with $E[Q_1(X)^2] = E[Q_2(Y)^2] \leq 4 t_{\rep}^2$, provides us an upper bound of the type
\begin{equation*}
\begin{split} P\left({\mathcal X}\times {\mathcal Y} \bigcap \mathcal{R} =\emptyset\right) &\leq 1 - 2 m_1 m_2 \frac{t_{\rep}}{t} + \\&+\frac{2 m_1 m_2 \left(m_1 m_2-1\right) t_{\text{rep}}^2}{t^2} + o((\trep/t)^2)
\end{split}
\end{equation*}
From this, we can estimate the gap between upper and lower bounds, and obtain the same asymptotic result as in Cor. \ref{cor:asymptotic21}.
\end{exe}
A more general upper bound can be found in \cite{CampVai2}. However the upper bound is not optimal, in the sense that it does not collapse with the lower bound for small $\trep$.

\end{appendices}
\bibliographystyle{plain}
\bibliography{storage}

\end{document}